\documentclass[lettersize,journal]{IEEEtran}

\usepackage[english]{babel}
\usepackage[noadjust]{cite}
\usepackage{algorithm}
\usepackage{algpseudocode}
\usepackage{optidef}

\ifCLASSINFOpdf
\else
   \usepackage[dvips]{graphicx}
\fi
\usepackage{url}

\usepackage{multirow}
\usepackage{tabularx}
\usepackage{graphicx}
\usepackage{optidef}
\usepackage{amsmath}
\usepackage{bm}
\usepackage{amssymb}
\usepackage{xcolor}
\usepackage{subcaption}
\usepackage{lipsum}
\usepackage{amsthm}
\usepackage{siunitx}
\usepackage{booktabs}
\usepackage{makecell}
\usepackage[acronym]{glossaries}

\usepackage[hidelinks]{hyperref}
\usepackage[acronym,automake]{glossaries-extra}

\usepackage{titlesec}
\titlespacing\section{0pt}{*1}{*1}
\titlespacing\subsection{0pt}{*1}{*1}
\titlespacing\subsubsection{0pt}{*0.8}{*0.8}

\newtheorem{theorem}{Theorem}
\newtheorem{corollary}{Corollary}[theorem]

\makeatletter
\newrobustcmd*{\myGls}{\@gls@hyp@opt\@myGls}
\newcommand*{\@myGls}[2][]{%
  \new@ifnextchar[{\@myGls@{#1}{#2}}{\@myGls@{#1}{#2}[]}%
}
\def\@myGls@#1#2[#3]{%
  \glsdoifexists{#2}%
  {%
    \let\do@gls@link@checkfirsthyper\@gls@link@checkfirsthyper
    \let\glsifplural\@secondoftwo
    \let\glscapscase\@secondofthree
    \def\glscustomtext{%
    \ifglsused{#2}
    {\acronymfont\glsentryshort{#2}#3}
    {%
    \ecapitalisewords{\glsentrylong{#2}}#3\space%
    \firstacronymfont(\glsentryshort{#2})}%
    }%
    \def\glsinsert{#3}%
    \def\@glo@text{\csname gls@\glstype @entryfmt\endcsname}%
    \@gls@link[#1]{#2}{\@glo@text}%
    \ifKV@glslink@local
      \glslocalunset{#2}%
    \else
      \glsunset{#2}%
    \fi
  }%
  \glspostlinkhook
}

\newrobustcmd*{\myGlspl}{\@gls@hyp@opt\@myGlspl}
\newcommand*{\@myGlspl}[2][]{%
  \new@ifnextchar[{\@myGlspl@{#1}{#2}}{\@myGlspl@{#1}{#2}[]}%
}
\def\@myGlspl@#1#2[#3]{%
  \glsdoifexists{#2}%
  {%
    \let\do@gls@link@checkfirsthyper\@gls@link@checkfirsthyper
    \let\glsifplural\@firstoftwo
    \let\glscapscase\@secondofthree
    \def\glscustomtext{%
    \ifglsused{#2}%
    {\acronymfont\glsentryshortpl{#2}#3}
    {%
    \ecapitalisewords{\glsentrylongpl{#2}}#3\space%
    \firstacronymfont(\glsentryshortpl{#2})}%
    }%
    \def\glsinsert{#3}%
    \def\@glo@text{\csname gls@\glstype @entryfmt\endcsname}%
    \@gls@link[#1]{#2}{\@glo@text}%
    \ifKV@glslink@local
      \glslocalunset{#2}%
    \else
      \glsunset{#2}%
    \fi
  }%
  \glspostlinkhook
}
\makeatother

\setabbreviationstyle[acronym]{long-short}

\makeglossaries

\newacronym{psr}{PSR}{power split ratio}
\newacronym{isac}{ISAC}{integrated sensing and communications}
\newacronym{mimo}{MIMO}{multiple-input multiple-output}
\newacronym{iv}{IV}{instrumental variable}
\newacronym{nsp}{NSP}{null-space projection}
\newacronym{dof}{DoF}{degree of freedom}
\newacronym{sinr}{SINR}{signal-to-interference-plus-noise ratio}
\newacronym{sdma}{SDMA}{spatial division multiplexing}
\newacronym{csi}{CSI}{channel state information}
\newacronym{mui}{MUI}{multi-user interference}
\newacronym{rcg}{RCG}{Riemannian conjugate gradient}
\newacronym{crb}{CRB}{Cramer-Rao bound}
\newacronym{lmmse}{LMMSE}{linear minimum mean square error}
\newacronym{isl}{ISL}{integrated sidelobe level}
\newacronym{papr}{PAPR}{peak-to-average power ratio}
\newacronym{swipt}{SWIPT}{simultaneous information and power transfer}
\newacronym{qos}{QoS}{quality-of-service}
\newacronym{sar}{SAR}{synthetic aperture radar}
\newacronym{svd}{SVD}{singular value decomposition}
\newacronym{xr}{XR}{extended reality}
\newacronym{3gpp}{3GPP}{Third Generation Partnership Project}
\newacronym{nr}{NR}{New Radio}
\newacronym{prs}{PRS}{Positioning Reference Signal}
\newacronym{etsi}{ETSI}{European Telecommunication Standards Institution}
\newacronym{kkt}{KKT}{Karush-Kuhn-Tucker}
\newacronym{mm}{MM}{majorization-minimization}
\newacronym{mse}{MSE}{mean square error}
\newacronym{sdr}{SDR}{semi-definite relaxation}
\newacronym{itu}{ITU}{International Telecommunication Union}
\newacronym{sdisac}{SD-ISAC}{spatial-division ISAC}
\newacronym{ofdm}{OFDM}{orthogonal frequency division multiplexing}

\newcommand{\BHL}[1]{\textcolor{black}{#1}}

\begin{document}
\title{Spatial-Division ISAC: A Practical Waveform Design Strategy via Null-Space Superimposition}

\author{Byunghyun Lee, \IEEEmembership{Graduate Student Member, IEEE}, Hwanjin Kim, \IEEEmembership{Member, IEEE}, David J. Love, \IEEEmembership{Fellow, IEEE}, and James V. Krogmeier, \IEEEmembership{Senior Member, IEEE}
\thanks{
A preliminary version of this work has been accepted for publication in the Asilomar Conference on Signals, Systems, and Computers, 2024  \cite{lee2024ISAC}.}
\thanks{
This work is supported in part by the National Science Foundation under grants EEC-1941529, CNS-2212565, and CNS-2225578 and the Office of Naval Research under grant N000142112472.
}
\thanks{Byunghyun Lee, David J. Love, and James V. Krogmeier are with the Elmore Family School of Electrical and Computer
Engineering, Purdue University, West Lafayette, IN 47907 USA (e-mails:
\{lee4093,djlove,jvk\}@purdue.edu).
}
\thanks{Hwanjin Kim is with the School of Electronics Engineering, Kyungpook National University, Daegu 41566, South Korea (e-mail:
hwanjin@knu.ac.kr).
}
}

\maketitle

\begin{abstract}

\Gls{isac} is a key enabler of new applications, such as precision agriculture, extended reality (XR), and digital twins, for 6G wireless systems.
However, the implementation of \gls{isac} technology is very challenging due to practical constraints such as high complexity.
In this paper, we introduce a novel \gls{isac} waveform design strategy, called \textit{the \gls{sdisac} waveform}, which simplifies the \gls{isac} waveform design problem by decoupling it into separate communication and radar waveform design tasks. 
Specifically, the proposed strategy leverages the null-space of the communication channel to superimpose sensing signals onto communication signals without interference.
This approach offers multiple benefits, including reduced complexity and the reuse of existing communication and radar waveforms.
We then address the problem of optimizing the spatial and temporal properties of the proposed waveform.
We develop a low-complexity beampattern matching algorithm, leveraging a \gls{mm} technique.
Furthermore, we develop a range sidelobe suppression algorithm based on manifold optimization.
We provide comprehensive discussions on the practical advantages and potential challenges of the proposed method, including null-space feedback. 
We evaluate the performance of the proposed waveform design algorithm through extensive simulations.
Simulation results show that the proposed method can provide similar or even superior performance to existing \gls{isac} algorithms while reducing computation time significantly.

\end{abstract}

\begin{IEEEkeywords}
Massive multiple-input multiple-output (MIMO), integrated sensing and communication (ISAC)
\end{IEEEkeywords}

\IEEEpeerreviewmaketitle
\glsresetall



\section{Introduction}

\Gls{isac} is a key technology of sixth generation (6G) wireless systems and will enable emerging applications like precision agriculture, \gls{xr} and digital twins \cite{brinton2024key6g,zhangTerahertzIntegratedSensing2025}.
The primary goal of \gls{isac} is to equip wireless systems with sensing capabilities (e.g., detection, tracking, imaging, mapping) through a co-design of communication and sensing on the shared infrastructure and spectrum.
 Recognizing its growing significance, the \gls{itu} included \gls{isac} as a key aspect of IMT-2030 \cite{ITU_IMT_2030_Framework_2023}.
\BHL{ 
Along the same line, \gls{3gpp} and \gls{etsi} have also initiated studies and standardization efforts on \gls{isac} \cite{3gpp.22.837,kaushik2024toward}.
}

The most straightforward way to realize \gls{isac} is to use a dedicated signal and resource for sensing, such as the \gls{prs} in 5G \gls{nr} \cite{3gpp.38.211}.
However, this approach may not be optimal in terms of resource efficiency given the expected spectral congestion from the heavy traffic predicted in 6G networks. 
More recently, joint signaling of radar and communications on shared time and frequency resource has received a great deal of interest due to its potential to improve sensing-communication trade-offs \cite{liuJointTransmitBeamforming2020a,zhangIntegratedSensingCommunication2024,liuDualFunctionalRadarCommunicationWaveform2021,lee2024constant,liuRangeSidelobeReduction2020,liuCramerRaoBoundOptimization2022}.
In joint \gls{isac} signal design, it is critical to consider the fundamental difference between communications and sensing.
From a communications perspective, sensing signals are treated as interference to be eliminated.
By contrast, a sensing perspective seeks to reuse communication signals to maximize target illumination and enhance sensing quality.


The design of a joint \gls{isac} signal is inherently more complex than the design of a communication-only or radar-only signal since the joint design must meet the requirements of both functions, such as minimum rate, sidelobe level, and peak-to-average power ratio \cite{zhouIntegratedSensingCommunication2022a}.
These constraints complicate both the mathematical formulation and solution algorithms. 
Furthermore, in joint \gls{isac} signaling, the radar and communication signals are closely coupled.
This approach enforces simultaneous updates of radar and communication signals, potentially resulting in overly frequent signal updates and increased feedback overhead.
For example, while the communication precoder must be updated according to the channel coherence time, the radar waveform follows a different cycle determined by the coherent processing interval (CPI).
This tight coupling also limits the reuse of conventional precoding schemes from legacy communication systems, creating the need for system upgrades and standardization changes.


In this paper, we introduce a novel ISAC waveform design strategy, referred to as \textit{the \gls{sdisac} waveform}, to address these challenges. 
The proposed method is inspired by the \gls{nsp} and \gls{iv} techniques.
The \gls{nsp} method has been utilized to manage interference in the radar-communication coexistence scenario \cite{sodagariProjectionBasedApproach2012a,mahalSpectralCoexistenceMIMO2017,khawarTargetDetectionPerformance2015}.
The \gls{nsp} method projects the radar signal onto the null space of the communication channel to avoid interference.
In \cite{liuMUMIMOCommunicationsMIMO2018a}, the authors extended the \gls{nsp} approach to a colocated radar-communication scenario where communication and radar transmit arrays are colocated but not shared.
\begin{table*}[!t]
  \centering
  \renewcommand{\arraystretch}{1.4}
   \caption{Comparison of the proposed SD-ISAC method and existing schemes.}
  \begin{tabular}{lcccccc}
    \toprule
    Reference & Scenario & \makecell{Design\\ Approach} & \makecell{Null-Space\\ Projection} & \makecell{Comm.-Sensing\\Separability} & \makecell{Compatibility w.\\Existing Precoding} & \makecell{Radar Algorithm\\Reusability} \\ 
    \midrule \midrule
    \cite{liuJointTransmitBeamforming2020a} & \gls{isac} & \makecell{Beampattern Shaping}& - & - & - & - \\ 
    \cite{zhangIntegratedSensingCommunication2024} & \gls{isac} & \makecell{Joint Channel and\\ Target Parameter Estimation}& - & - & - & - \\ 
     \cite{liuDualFunctionalRadarCommunicationWaveform2021} & \gls{isac} & \makecell{Beampattern Shaping} & - & - & - & - \\ 
    \cite{liuRangeSidelobeReduction2020}   & \gls{isac} & \makecell{Correlation Suppression} & - & - & - & - \\ 
    \cite{liuCramerRaoBoundOptimization2022}   & \gls{isac} & \makecell{Cramer-Rao Bound\\ Minimization }& - & - & - & - \\ \midrule
    \cite{sodagariProjectionBasedApproach2012a} & \makecell{Radar-Comm. \\ Coexistence}  & \makecell{Interference Nulling} & \checkmark & \checkmark & \checkmark & \checkmark \\ 
\cite{mahalSpectralCoexistenceMIMO2017} & \makecell{Radar-Comm. \\ Coexistence}  &\makecell{Coordinated  Multi-point } & \checkmark & \checkmark & \checkmark & \checkmark \\ 
\cite{khawarTargetDetectionPerformance2015} & \makecell{Radar-Comm. \\ Coexistence}  &\makecell{Projection Matrix Selection } & \checkmark & \checkmark & \checkmark & \checkmark \\ 
    \cite{liuMUMIMOCommunicationsMIMO2018a}   & \makecell{Colocated \\ Radar-Comm.}  &\makecell{Beampattern Shaping}  & \checkmark & \checkmark & \checkmark & - \\ 
    \midrule \midrule
    \makecell{\textbf{SD-ISAC}\\ (Proposed)}
     & \gls{isac} & \makecell{Applicable to\\ any approach} & \checkmark & \checkmark & \checkmark & \checkmark \\
    \bottomrule
  \end{tabular}
 
  \label{tab:literature}
\end{table*}
While the \gls{nsp} technique is a powerful tool in managing interference between radar and communications, it is not applicable to \gls{isac} systems where the transmit array is fully shared between radar and communications.
The \gls{iv} method has been widely adopted in time-series analysis, system identification, and array processing  \cite{mosesInstrumentalVariableAdaptive1988,Torsten1989System,stoicaInstrumentalVariableApproach1994}.
In \cite{liSignalSynthesisReceiver2008,huaReceiverDesignRange2013a}, an \gls{iv} filter was proposed for range and Doppler sidelobe reduction of \gls{mimo} radar signals as a replacement for the matched filter.
The core idea of the IV filter is to design \textit{an added variable} orthogonal to the transmit signal so that sidelobe suppression can be done without changing the mainlobe response.
Despite its potential, the use of the \gls{iv} method has been limited to receiver-side solutions.

Inspired by this, we propose a novel \gls{isac} waveform design strategy based on the \gls{nsp} and \gls{iv} methods.
The proposed waveform, referred to as \textit{the \gls{sdisac} waveform}, is structured with two components: a communication signal and an added sensing signal.
The basic rationale is that the sensing signal is constructed on the null space of the communication channel, to optimize radar metrics without interference.
By doing so, the proposed approach can turn the ISAC waveform design problem into separate communication-only and radar-only waveform design tasks, offering advantages such as reduced complexity.
The proposed method does not limit the communication signal to any specific precoding scheme, improving the compatibility with the existing precoders and mitigating the need to change legacy systems, including infrastructure and communication standards. 
We demonstrate optimizing the spatial and temporal properties of the proposed waveform.
Simulation results show that the proposed method outperforms the existing ISAC algorithm while offering additional benefits.



\BHL{The comparison of our approach with existing approaches is provided in
Table \ref{tab:literature}, and the main contributions of this paper can be summarized as:}
\begin{itemize}
    \item 
    We introduce a novel \gls{isac} waveform structure, called \textit{the \gls{sdisac} waveform}.
    The proposed method superimposes sensing signals onto communication signals without interfering with the communication channel.
    The proposed method provides various advantages, such as low complexity and compatibility with standard precoding techniques.
    These advantages are essential, as they help maximize the reuse of legacy systems, including infrastructure and communication standards.
    \item 
    We optimize the spatial beampattern of the proposed waveform.
    Thanks to the separability of the proposed method, the beampattern matching problem can be formulated as radar-only optimization without communication constraints, which admits a simple solution.
    In particular, we revisit the beampattern matching cost function introduced in \cite{fuhrmannTransmitBeamformingMIMO2004}, which provides superior beam shaping performance while simplifying the formulation.
    We leverage a \gls{mm} technique to approximate the objective with a quadratic function. 
    The approximated problem is then simplified into single-variable optimization using the \gls{kkt} conditions.
    
    
    \item 
    We present the problem of suppressing the range sidelobes of the proposed waveform.
    Similarly, range sidelobe suppression can be formulated as a radar-only problem, facilitating the reuse of existing algorithms.   
    Specifically, we minimize \gls{isl} of the waveform using a first-order Riemannian manifold optimization method.
    \item 
    We discuss the benefits and potential issues of the proposed waveform, including null-space feedback.    We show that the proposed waveform can be applied to other similar problems, such as \gls{swipt} waveform design and \gls{papr} minimization.
    
    \item We evaluate the performance of the proposed method through extensive simulations.
    Simulation results show the proposed strategy can perform comparably or outperform existing \gls{isac} methods without jointly optimizing communication and sensing. 
\end{itemize}


The remainder of the paper is organized as follows. 
In Sec. II, we set up the system model along with the communication and radar models and propose our \gls{isac} waveform design strategy.
Then, in Sec. III and Sec. IV, we focus on optimizing the spatial and temporal properties of the proposed waveform, respectively.
In Sec. V, we extend our proposed scheme to the imperfect \gls{csi} scenario.
Next, in Sec. VI, we discuss the advantages and potential challenges.
In Sec. VII, we evaluate our proposed algorithms.
Finally, we conclude the paper in Sec. VIII.





\begin{figure*}[!t]
\center{\includegraphics[width=.87\linewidth]{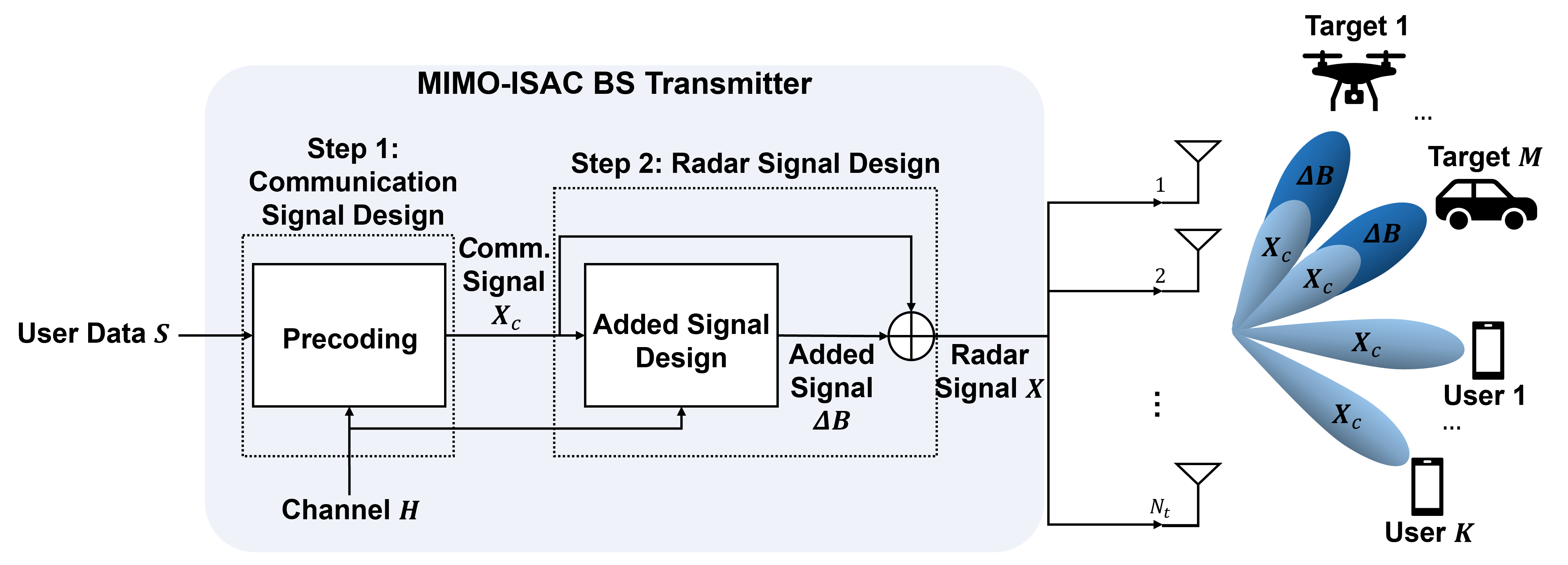}}
\caption{\small Proposed \gls{mimo}-\gls{isac} system.
The proposed waveform synthesis strategy consists of the following steps: (i) the communication signal is generated, and then (ii) the added signal $\bm{\Delta}\textbf{B}$ is constructed on the null-space of the communication channel and superimposed onto the communication signal $\textbf{X}_c$ without interference.
The resulting composite signal $\textbf{X}$ is used to illuminate the targets, while the communication component $\textbf{X}_c$ is responsible for conveying information bits.  
}
\label{fig:system}
\end{figure*}

\textbf{Notation}:
Vectors and matrices are denoted by boldface lowercase and uppercase letters, respectively.
$(\cdot)^T$, $(\cdot)^*$, $(\cdot)^H$, and $(\cdot)^{-1}$ are the transpose, conjugate, conjugate transpose, and inverse operators, respectively.
$|\cdot|$, $\Vert \cdot\Vert$, and $\Vert \cdot\Vert_F$ denote the absolute, 2-norm, and Frobenius norm operators, respectively.
$\text{Tr}(\cdot)$ and $\text{Rank}(\cdot)$ are the trace and rank of a matrix, respectively.
$\text{Re}\{\cdot\}$ returns the real part of a complex number/vector/matrix.
$\mathbb{E}[\cdot]$ is the expectation operator.
$[\cdot]_{i,j}$ denotes the $(i,j)$th entry of the input matrix.
$\textbf{0}$ and $\textbf{I}$ represent the all-zero and identity matrices, respectively.
$\textbf{A}\succeq \textbf{B}$ indicates $\textbf{A}-\textbf{B}$ is a positive semidefinite matrix.
$\langle\cdot,\cdot\rangle$ denotes the Euclidean inner product.

\section{System Model and Proposed Waveform}

\subsection{System Model}
Consider a \gls{mimo}-\gls{isac} system in which a single base station (BS) functions as a communication transmitter and monostatic\footnote{
Despite the considered monostatic scenario, the proposed \gls{sdisac} waveform can also be applied to bistatic or multistatic scenarios, provided that the transmit array is shared by sensing and communications. 
} radar concurrently, as depicted in Fig. \ref{fig:system}.
The BS is equipped with a transmit array of ${N_t}$ antennas for joint radar and communication transmission.
A radar receive array of $N_r$ antennas is colocated with the transmit array at the BS to capture reflected signals.
The BS transmits $N$ bit streams to $K$ communication antennas where $K\geq N$.
In the considered \gls{isac} system, the BS transmits a downlink transmit signal $\textbf{X}\in\mathbb{C}^{{N_t}\times L}$ of length $L$ that simultaneously illuminates radar targets and conveys information bits for communications.


We adopt a narrowband\footnote{\BHL{Although this paper considers the narrowband system as in many existing works in the literature \cite{liuJointTransmitBeamforming2020a,lee2024constant,liuCramerRaoBoundOptimization2022,liuRangeSidelobeReduction2020}, the proposed \gls{sdisac} can be extended to wideband systems. We provide discussions on this wideband extension in Section VI.}} block-fading channel model where the channel does not change for $L$ symbol times.
Given the transmit signal $\textbf{X}$, the input-output relationship for communications can be expressed as
\begin{equation}\label{eq:comm_signal}
    \textbf{Y}_c=\textbf{H}\textbf{X}+\textbf{W}_c,
\end{equation}
where 
$\textbf{H}=[\textbf{h}_1^H,\textbf{h}_2^H,\dots,\textbf{h}_K^H]^H\in\mathbb{C}^{K\times {N_t}}$ is the communication channel
with $\textbf{h}_k\in \mathbb{C}^{{N_t}\times 1}$ being the channel from the BS to user $k$
and $\textbf{W}_c\in\mathbb{C}^{K\times L}$ is the noise matrix with independent and identically distributed (i.i.d.) entries drawn from $\mathcal{CN}(0,\sigma_c^2)$. 
We assume that the BS has perfect knowledge of \textbf{H}.
Without loss of generality, for the rest of this paper, we focus on the $K$ single-antenna multi-user scenario where $N=K$ and the $k$th row of the channel matrix corresponds to the channel of user $k$, i.e., $\textbf{H}=[\textbf{h}_1,\textbf{h}_2,\dots,\textbf{h}_K]^H$ with $\textbf{h}_k\in \mathbb{C}^{{N_t}\times 1}$ being the channel from the BS to user $k$.


The transmitted signal is reflected at the target and clutter and received at the BS.
The received reflected signal at the radar receiver can be expressed as
\begin{equation}
    \begin{aligned}        
    \textbf{Y}_r
    =\displaystyle\sum_{m=1}^M\alpha_m\textbf{a}_r(\theta_m)\textbf{a}_t^H(\theta_m)\bar{\textbf{X}}\textbf{J}_{\tau_m} +\textbf{W}_r
    \end{aligned}
\end{equation}
where $\bar{\textbf{X}}=[\textbf{X},\textbf{0}_{{N_t}\times L_0}]\in \mathbb{C}^{{N_t}\times (L+L_0)}$ is a zero-padded transmit waveform,
$M$ is the number of targets, $\alpha_m\in\mathbb{C}$ is the radar cross-section of target $m$,
$\textbf{a}_t(\cdot)\in\mathbb{C}^{N_t}$ is the steering vector of the transmit array, 
$\textbf{a}_r(\cdot)\in\mathbb{C}^{N_r}$ is the steering vector of the radar receive array,
$\theta_m$ is the angle of target $m$, 
$\textbf{J}_{\tau_m}$ is the shift matrix, $\tau_m$ is the round-trip delay for target $m$,
and $\textbf{W}_r$ is the combined noise and clutter matrix.
The radar receiver stores $L+L_0$ received pulses to capture the delayed reflected signals. 
The shift matrix accounts for the round-trip delay of each target, and the $(i,j)$th entry of the shift matrix for delay $\tau$ is defined as \cite{horn2012matrix,he2009designing}
\begin{equation}\label{eq:shift_mat}
    [\textbf{J}_{\tau}]_{i,j} = 
    \begin{cases}
        1, & \text{if }j-i=\tau \\
        0, & \text{otherwise}
    \end{cases}
\end{equation}
and $\textbf{J}_{\tau}=\textbf{J}^T_{-\tau}$.
The target parameters such as angles and delays can be extracted from the received echo signal using receiver techniques like matched filtering.

\subsection{Proposed \gls{sdisac} Waveform}

In this section, we elaborate on the proposed waveform structure.
A typical \gls{isac} waveform design problem can be expressed as
\begin{mini}|s|
{\small\textbf{X}}{g_r(\textbf{X})}
{}{}
\addConstraint{g_{c}(\textbf{X})\leq 0}{}
\addConstraint{\Vert \textbf{X} \Vert_F^2 =LP_t}{}
\label{prob:DFRC}
\end{mini}
where $g_{r}(\textbf{X})$ is the radar cost function, $g_{c}(\textbf{X})$ is the communication constraint, and $P_t$ is the transmit power budget.
As mentioned, a common approach is to directly design the transmit signal $\textbf{X}$ 
\cite{liuDualFunctionalRadarCommunicationWaveform2021,lee2024constant,liuRangeSidelobeReduction2020}.
This approach is essentially nonlinear precoding, which requires using knowledge of the realization of the data symbols.
As an alternative approach, \cite{liuCramerRaoBoundOptimization2022,liuJointTransmitBeamforming2020a} design the sample covariance assuming that the sample and statistical covariance matrices are identical.
While this approximation holds true if the codeword length is infinite, this can introduce a non-negligible approximation error due to random data realization, especially for the massive MIMO scenario \cite{luRandomISACSignals2024a}.


To tackle the aforementioned challenges, we propose a novel waveform structure that can decompose a dual-function waveform design problem into two single-function waveform design problems.
The proposed signal structure can be expressed as
\begin{equation}\label{eq:Tx_Waveform}
\begin{aligned}
\textbf{X}=\underbrace{\textbf{F}\textbf{S}}_{\textbf{X}_c}+\bm{\Delta}\textbf{B}=\underbrace{\textbf{X}_c}_{\text{comm. signal}}+\underbrace{\bm{\Delta}\textbf{B}}_{\text{added signal}}
\end{aligned}
\end{equation}
where 
$\textbf{F}=[\textbf{f}_1,\textbf{f}_2,\dots,\textbf{f}_K]\in \mathbb{C}^{{N_t}\times K}$ is the communication precoding matrix with $\textbf{f}_k$ being the beamformer for user $k$ and $\textbf{S}=[\textbf{s}_1,\textbf{s}_2,\dots,\textbf{s}_K]^T\in \mathbb{C}^{K\times L}$ is the user data matrix with $\textbf{s}_k$ being the data for user $k$ and $\mathbb{E}[\textbf{S}\textbf{S}^H]=\textbf{I}_K$.
$\bm{\Delta}\in\mathbb{C}^{{N_t}\times ({N_t}-K)}$ is a semi-unitary matrix whose columns span the null-space of the channel matrix $\textbf{H}$
 with $\bm{\Delta}^H\bm{\Delta} = \textbf{I}_{{N_t}-K}$
, and $\textbf{B}\in\mathbb{C}^{({N_t}-K)\times L}$ is an arbitrary matrix.
The projection matrix $\bm{\Delta}$ can be obtained via the \gls{svd} of $\textbf{H}$, which is given by
\begin{equation}
    \textbf{H} = \textbf{U}\bm{\Sigma}\textbf{V}^H.
\end{equation}
The projection matrix $\bm{\Delta}$ contains the ${N_t}-K$ rightmost column vectors of $\textbf{V}=[\textbf{v}_1,\textbf{v}_2,\dots,\textbf{v}_{{N_t}}]$ corresponding to the zero singular values, which is given by 
\begin{equation}
    \bm{\Delta}=[\textbf{v}_{K+1},\dots,\textbf{v}_{N_t}].
\end{equation}
It follows that $\textbf{H}\bm{\Delta}=\textbf{0}$. 

\BHL{The added signal $\bm{\Delta}\textbf{B}$ contains discrete transmit symbols that are superimposed onto the communication transmit symbols to formulate the composite transmit symbols.
$[\textbf{X}]_{n,\ell}$ corresponds to the $\ell$th discrete transmit symbol at the $n$th transmit antenna, which can be converted to a continuous signal using a pulse-shaping filter.
Note that the added signal $\bm{\Delta}\textbf{B}$ does not introduce a change in the symbol rate or the bandwidth.
}

For clarity, we refer to the first and second terms in \eqref{eq:Tx_Waveform} as communication and additive signals, respectively.
By substituting \eqref{eq:Tx_Waveform} into \eqref{eq:comm_signal}, the communication input-output can be rewritten as 
\begin{align}
    \textbf{Y}_c&=\textbf{H}\textbf{X}+\textbf{W}_c=
    \textbf{H}(\textbf{X}_c+\bm{\Delta}\textbf{B})+\textbf{W}_c \\ \nonumber
    &=  \textbf{H}\textbf{X}_c+\textbf{W}_c.
\end{align}
It can be confirmed that the added signal always falls into the null space of the channel matrix, i.e., $\textbf{H}\bm{\Delta}=\textbf{0}_{K \times {N_t}-K}$, thereby  causing zero interference to the communication channel.


The main idea of the proposed method is that the communication signal $\textbf{X}_c$ delivers information bits, while the composite signal $\textbf{X}=\textbf{X}_c+\bm{\Delta}\textbf{B}$ illuminates targets.
This simple but powerful idea is based on the contrasting perspectives of communications and sensing.
Communications treats sensing signals as interference and, therefore, aims to eliminate them.
In contrast, sensing seeks to reuse the energy of communication signals to illuminate targets.
The proposed method satisfies the needs of the two functions through the superimposition on the null-space of the communication channel.

In the proposed framework, the transmit signal is constructed by designing matrices $\textbf{X}_c$ and $\textbf{B}$ individually, which significantly simplifies problem formulation. 
It is important to note $\textbf{X}_c$ can be formulated via traditional precoding techniques from linear precoding (e.g., zero-forcing, codebook-based beamforming) to nonlinear precoding. 
After formulating $\textbf{X}_c$, the added signal $\textbf{B}$ is constructed on top of $\textbf{X}_c$. In the following, we will focus on designing $\textbf{B}$ to optimize the waveform beam pattern for radar sensing.

The major difference between the \gls{sdisac} waveform and the \gls{iv} filter \cite{liSignalSynthesisReceiver2008,huaReceiverDesignRange2013a} is the power constraint.
The added signal uses the remaining transmit power after constructing the communication signal.
Denote the communication signal power and added signal power by $P_c$ and $P_a$, respectively.
The communication signal power can be written as $P_c=\mathbb{E}[\Vert \textbf{F}\textbf{S}\Vert_F^2]/L=\mathbb{E}[\text{Tr}(\textbf{F}\textbf{S}\textbf{S}^H\textbf{F}^H)]/L=\text{Tr}(\textbf{F}\mathbb{E}[\textbf{S}\textbf{S}^H]\textbf{F}^H)/L=\Vert \textbf{F}\Vert_F^2/L$ where the last equality follows from $\mathbb{E}[\textbf{S}\textbf{S}^H]=\textbf{I}_K$.
Moreover, given $\bm{\Delta}^H\bm{\Delta}=\textbf{I}_{{N_t}-K}$, we have $\Vert\bm{\Delta}\textbf{B}\Vert_F^2 = \text{Tr}\left(\textbf{B}^H\bm{\Delta}^H\bm{\Delta}\textbf{B}\right)= \Vert \textbf{B} \Vert _F ^2 = LP_a$.
The total transmit power can be expressed as $\mathbb{E}[\Vert\textbf{X}\Vert_F^2] = LP_t =L(P_c + P_a)$.
We define the ratio of the communication signal power to the total power, i.e., $\beta=P_c/P_t$, as the power split ratio.
The power split ratio controls the trade-off between communications and sensing.
A higher power split ratio value corresponds to the higher power of the communication signal, but less power available for the added signal.

\BHL{To determine the power split ratio, we first identify the communication power required to meet the minimum user \gls{qos}.
Without loss of generality, this paper considers power minimization under a per-user \gls{sinr} constraint for designing communication signals, which is formulated as 
\begin{equation}
    \begin{aligned}
        & \underset{\small\textbf{F}}{\min}
        & & \Vert \textbf{F} \Vert_F^2 \\
        & \text{s.t.}
        & & \frac{|\textbf{h}^H_k \textbf{f}_k|^2}{\sum_{k' \neq k}^K |\textbf{h}^H_k \textbf{f}_{k'}|^2 + L\sigma_c^2} \geq \gamma_k, \quad \forall k
    \end{aligned}
    \label{prob:comm}
\end{equation}
where $\gamma_k$ is the SINR threshold for user $k$.
By solving the problem \eqref{prob:comm}, we aim to assign the minimum power to communications satisfying the user \gls{sinr} requirements, maximizing the residual power for sensing.}


To summarize, the proposed waveform design framework consists of the following steps.
\begin{enumerate}
    \item Formulate the beamforming matrix $\textbf{F}$ and communication signal $\textbf{X}_c=\textbf{F}\textbf{S}$.
    \item Compute the \gls{svd} of the communication channel and the projection matrix $\bm{\Delta}$.
    \item Given the communication signal, optimize the radar cost function over the added signal $\textbf{B}$. 
    Then, formulate the transmit waveform by $\textbf{X}=\textbf{F}\textbf{S}+\bm{\Delta}\textbf{B}$.
\end{enumerate}
The above process is summarized in Fig. \ref{fig:system}.
In the following sections, we will demonstrate synthesizing the matrix $\textbf{B}$ with the goal of optimizing the spatial and temporal properties of the proposed waveform.

\section{Application 1: Beampattern Optimization}

In this section, we address the design of the spatial beampattern for the proposed waveform.
In radar beampattern design, it is critical to maximize mainlobes for strong target reflections while minimizing sidelobes for clutter suppression.
\BHL{
There are several available metrics for beampattern design such as the beampattern \gls{mse} \cite{stoica2007probing,aittomakiLowComplexityMethodTransmit2007,ahmedUnconstrainedSynthesisCovariance2011,ahmedMIMORadarTransmit2014}, minimum peak sidelobe level ratio \cite{fanConstantModulusMIMO2018a}, minimum integrated sidelobe level ratio \cite{aubryMIMORadarBeampattern2016c} and \gls{crb} \cite{liRangeCompressionWaveform2008}. 
A beampattern matching approach aims to align the beampattern of the waveform with a reference beampattern with desired properties, such as a low sidelobe level \cite{li2008mimo}.
Although the proposed \gls{sdisac} waveform is compatible with any design, we focus on the beampattern matching approach due to its applicability to various scenarios such as multi-target scenarios.
}
We develop a low-complexity beampattern matching algorithm based on a majorization-minimization (MM) technique.
Specifically, we will find a surrogate function for the nonconvex objective function and minimize it via a simple line search.

\subsection{Problem Formulation}

The beampattern at angle $\theta$ is given by \cite{li2008mimo}
\begin{equation}   \label{eq:beampattern}
\begin{aligned}    
G(\textbf{B},\theta)
&=\Vert \textbf{X}^H\textbf{a}_t(\theta) \Vert^2 \\
&= \Vert \textbf{B}^H{\bm{\Delta}^H\textbf{a}_t(\theta)} + {\textbf{X}_c^H\textbf{a}_t(\theta)}  \Vert^2 \\
&= \Vert \textbf{B}^H\bar{\textbf{a}}+\textbf{r}\Vert^2
\end{aligned}
\end{equation}
where $\bar{\textbf{a}}=\bm{\Delta}^H\textbf{a}_t(\theta)$ and $\textbf{r}=\textbf{X}_c^H\textbf{a}_t(\theta)$.
Several cost functions have been proposed for beampattern matching 
\cite{stoica2007probing,aittomakiLowComplexityMethodTransmit2007,ahmedMIMORadarTransmit2014}, and the choice of cost function significantly impacts performance. 
In this paper, we focus on minimizing the \gls{mse} between the square roots of the designed and desired beampatterns \cite{fuhrmannTransmitBeamformingMIMO2004}. 
With this in mind, the beampattern matching cost is defined as \cite{fuhrmannTransmitBeamformingMIMO2004}
\begin{equation}\label{eq:beampattern_cost0}
\begin{aligned}
    g(\textbf{B})&=\displaystyle\sum_{u=1}^U\left|  \sqrt{G(\textbf{B},\theta_u)}-\sqrt{d(\theta_u)}\right|^2  \\
    &=\displaystyle\sum_{u=1}^U\left|  \left\Vert  \textbf{B}^H\bar{\textbf{a}}_u+\textbf{r}_u\right\Vert-\sqrt{d(\theta_u)}\right|^2
\end{aligned}
\end{equation}
where $U$ is the number of discretized angles, $u$ is the angle index,
 and $d(\theta_u)\geq 0$ is the reference beampattern at angle $\theta_u$.
The cost function \eqref{eq:beampattern_cost0} can be rewritten as
\begin{equation}\label{eq:beampattern_cost}
\begin{aligned}
    g(\textbf{B})
    &=\left\Vert  \textbf{B}^H\textbf{A}+\textbf{Z}\right\Vert_F^2- 2\displaystyle\sum_{u=1}^U   \sqrt{d(\theta_u)}\left\Vert  \textbf{B}^H\bar{\textbf{a}}+{\textbf{r}}_u\right\Vert+\bar{d}
\end{aligned}   
\end{equation}
where 
\begin{equation*}
    \begin{aligned}
        \textbf{A} &= [\bar{\textbf{a}}_1,\bar{\textbf{a}}_2,\dots,\bar{\textbf{a}}_U], \
        \textbf{Z} = [\textbf{r}_1,\textbf{r}_2,\dots,\textbf{r}_U] \ \text{and} \
        \bar{d} = \displaystyle\sum_{u=1}^U d(\theta_u).
    \end{aligned}
\end{equation*}

Given the formulated cost function and power constraint, we formalize a beampattern matching problem as
\begin{mini}|s| 
{\small\textbf{B}}{g(\textbf{B})} 
{}{} 
 \addConstraint{\Vert \textbf{B}\Vert_F^2=LP_a}{}.
 \label{prob:beampattern}
\end{mini}
The proposed \gls{sdisac} method allows \eqref{prob:beampattern} to be free from communication constraints, which significantly simplifies the solution.
Yet, directly minimizing the objective is still challenging due to the combination of norm and squared norm in \eqref{eq:beampattern_cost}.
To address this, we majorize \eqref{eq:beampattern_cost} with a more tractable function that admits a simple solution.

\subsection{Solution Algorithm: MM-LineSearch}

We begin by majorizing the second term in \eqref{eq:beampattern_cost}.
\begin{theorem}\label{theorem:majorizer}
    The second term in \eqref{eq:beampattern_cost} is majorized by a linear function with respect to $\textbf{B}$ at point $\textbf{B}_t$ as 
    \begin{equation}\label{eq:beamapttern_norm}
- 2\displaystyle\sum_{u=1}^U   \sqrt{d(\theta_u)}\left\Vert  \textbf{B}^H\bar{\textbf{a}}_u+{\textbf{r}}_u\right\Vert \leq      -\text{Re}\left\{\text{Tr}\left(\textbf{B}^H\textbf{D}_t\right)\right\}+ const
    \end{equation}
    where 
    \begin{equation}
        \begin{aligned}
        \textbf{D}_t= 
        \displaystyle\sum_{u=1}^U\frac{2\sqrt{d(\theta_u)}\bar{\textbf{a}}_u(\bar{\textbf{a}}_u^H \textbf{B}_t+\textbf{r}^H_u)}{\left\Vert\textbf{B}_t^H\bar{\textbf{a}}_u+{\textbf{r}}_u\right\Vert}
        \end{aligned}
    \end{equation}
  and the equality holds if $\textbf{B}=\textbf{B}_t$.
\end{theorem}
\begin{proof}
By the Cauchy-Schwarz inequality, it follows that \cite{sun2016majorization}
\begin{equation}\label{ineq:Cauchy}
    \left\Vert\textbf{B}^H\bar{\textbf{a}}_u+{\textbf{r}}_u\right\Vert \geq
    \frac{\text{Re}\left\{(\bar{\textbf{a}}_u^H \textbf{B}_t+\textbf{r}^H_u)(\textbf{B}^H\bar{\textbf{a}}_u+\textbf{r}_u)\right\}}{\left\Vert\textbf{B}_t^H\bar{\textbf{a}}_u+{\textbf{r}}_u\right\Vert}.
\end{equation}
Rearranging the left-hand side of \eqref{eq:beamapttern_norm} using \eqref{ineq:Cauchy} yields
\begin{equation}
\begin{aligned}
    - 2&\displaystyle\sum_{u=1}^U   \sqrt{d(\theta_u)}\left\Vert  \textbf{B}^H\bar{\textbf{a}}_u+{\textbf{r}}_u\right\Vert \\ &\leq 
   - 2\displaystyle\sum_{u=1}^U   
   \frac{\sqrt{d(\theta_u)}\text{Re}\left\{(\bar{\textbf{a}}_u^H \textbf{B}_t+\textbf{r}_u^H)(\textbf{B}^H\bar{\textbf{a}}_u+\textbf{r}_u)\right\}}{\left\Vert\textbf{B}_t^H\bar{\textbf{a}}_u+{\textbf{r}}_u\right\Vert} \\
   &=
   - 2\displaystyle\sum_{u=1}^U   
   \frac{\sqrt{d(\theta_u)}\text{Re}\left\{\text{Tr}\left((\textbf{B}^H\bar{\textbf{a}}_u+\textbf{r}_u)(\bar{\textbf{a}}_u^H \textbf{B}_t+\textbf{r}^H_u)\right)\right\}}{\left\Vert\textbf{B}_t^H\bar{\textbf{a}}_u+{\textbf{r}}_u\right\Vert}\\
   &=
   - 2\displaystyle\sum_{u=1}^U   
   \frac{\sqrt{d(\theta_u)}\text{Re}\left\{\text{Tr}\left(\textbf{B}^H\bar{\textbf{a}}_u(\bar{\textbf{a}}_u^H \textbf{B}_t+\textbf{r}^H_u)\right)\right\}}{\left\Vert\textbf{B}_t^H\bar{\textbf{a}}_u+{\textbf{r}}_u\right\Vert} +const\\
   &=
   -    
   \text{Re}\left\{\text{Tr}\left(\textbf{B}^H\displaystyle\sum_{u=1}^U\frac{2\sqrt{d(\theta_u)}\bar{\textbf{a}}_u(\bar{\textbf{a}}_u^H \textbf{B}_t+\textbf{r}^H_u)}{\left\Vert\textbf{B}_t^H\bar{\textbf{a}}_u+{\textbf{r}}_u\right\Vert}\right)\right\} +const\\
   &= -\text{Re}\left\{\text{Tr}\left(\textbf{B}^H\textbf{D}_t\right)\right\}+const
\end{aligned}
\end{equation}
where the first equality comes from the cyclic property of the trace operator and the constant term collects all terms irrelevant to $\textbf{B}$.
If $\textbf{B}=\textbf{B}_t$, then $\text{Re}\left\{(\bar{\textbf{a}}_u^H \textbf{B}_t+\textbf{r}_u^H)(\textbf{B}_t^H\bar{\textbf{a}}_u+\textbf{r}_u)\right\}=\left\Vert\textbf{B}_t^H\bar{\textbf{a}}_u+{\textbf{r}}_u\right\Vert^2$.
Thus, the equality holds if $\textbf{B}=\textbf{B}_t$,
which completes the proof.

\end{proof}
A majorizer of \eqref{eq:beampattern_cost} can be easily found using Theorem \ref{theorem:majorizer}.
\begin{corollary}
    The cost function \eqref{eq:beampattern_cost} can be majorized as
    \begin{equation}
      g(\textbf{B})  \leq \left\Vert  \textbf{B}^H\textbf{A}+\textbf{Z}\right\Vert_F^2 -  \text{Re}\left\{\text{Tr}\left(\textbf{B}^H\textbf{D}_t\right)\right\} + const=\tilde{g}(\textbf{B},\textbf{B}_t)
    \end{equation}
    where the equality holds if $\textbf{B}=\textbf{B}_t$.
\end{corollary}

The obtained majorizer $\tilde{g}(\textbf{B},\textbf{B}_t)$ is a quadratic function.
Ignoring the constant in $\tilde{g}(\textbf{B},\textbf{B}_t)$, the beampattern matching problem can be recast as
\begin{mini}|s|
{\small\textbf{B}}{\left\Vert  \textbf{B}^H\textbf{A}+\textbf{Z}\right\Vert_F^2 -  \text{Re}\left\{\text{Tr}\left(\textbf{B}^H\textbf{D}_t\right)\right\}}
{}{}
\addConstraint{\left\Vert\textbf{B} \right\Vert_F^2 =LP_a}{}.
\label{prob:beampattern_approx}
\end{mini}
The approximated problem \eqref{prob:beampattern_approx} can be seen as a well-known trust-region subproblem for which the strong duality holds \cite{fortin2004trust}.
It is shown that this class of problems can be solved via a simple 1-dimensional search \cite{liuDualfunctionalRadarCommunicationSystems2018}.
Following the same strategy, we use the \gls{kkt} conditions to find the optimal minimizer, which are given by
\begin{align}
    \nabla\mathcal{L}(\textbf{B},\lambda) &= (\textbf{A}\textbf{A}^H+\lambda \textbf{I})\textbf{B}-\textbf{G}_t=\textbf{0} \label{eq:stationary_condition}\\
    \Vert \textbf{B} \Vert_F^2 &= LP_a \label{eq:power_constraint}
\end{align}
where $\textbf{G}_t = \textbf{D}_t-\textbf{A}\textbf{Z}^H$.

From \eqref{eq:stationary_condition}, the optimal solution for \eqref{prob:beampattern_approx} can be obtained as 
\begin{equation}\label{eq:opt_B}
    \textbf{B}_{opt}= \left(\textbf{A}\textbf{A}^H+\lambda \textbf{I}\right)^{-1}\textbf{G}_t.
\end{equation}
By substituting \eqref{eq:opt_B} into \eqref{eq:power_constraint}, we have
\begin{equation}\label{eq:B_norm}
    \left\Vert \left(\textbf{A}\textbf{A}^H+\lambda \textbf{I}\right)^{-1}\textbf{G}_t \right\Vert_F^2 = LP_a.
\end{equation}
Let $\textbf{A}\textbf{A}^H=\textbf{Q}\bm{\Lambda}\textbf{Q}^H$ be the eigenvalue decomposition of $\textbf{A}\textbf{A}^H$.
With this, \eqref{eq:B_norm} can be rewritten as
\begin{equation}\label{eq:B_norm2}
\begin{aligned}
    \left\Vert \left (\textbf{Q}\bm{\Lambda}\textbf{Q}^H+\lambda \textbf{I}\right)^{-1}\textbf{G}_t \right\Vert_F^2 
    &= \left\Vert \textbf{Q}\left(\bm{\Lambda}+\lambda \textbf{I}\right)^{-1}\textbf{Q}^H\textbf{G}_t \right\Vert_F^2 \\
    &= \left\Vert \left(\bm{\Lambda}+\lambda \textbf{I}\right)^{-1}\textbf{Q}^H\textbf{G}_t \right\Vert_F^2
\end{aligned}    
\end{equation}
where the last equality follows from the fact that $\textbf{Q}$ is an orthogonal matrix.
Given that $\bm{\Lambda}$ is a diagonal matrix, \eqref{eq:B_norm2} can be further simplified as
\begin{equation}\label{eq:B_norm3}    P(\lambda)=\displaystyle\sum_{i=1}^{{N_t}-K}\displaystyle\sum_{j=1}^L\frac{\left|\left[\textbf{Q}^H\textbf{G}_t\right]_{i,j}\right|^2}{(\lambda_i+\lambda)^2}
\end{equation}
where $\lambda_i$ is the $i$th diagonal entry of $\bm{\Lambda}$.
It can be verified that, for $\lambda > -\min_i \lambda_{i}$, $P(\lambda)$ decreases monotonically.
Hence, the Lagrange dual variable $\lambda$ such that $P(\lambda)=LP_a$ can be found via a line search algorithm (e.g., bisection algorithm, backtracking algorithm).
Once the optimal $\lambda$ is found, the added signal can be updated as $\textbf{B}_{t+1}=\textbf{Q}\left(\bm{\Lambda}+\lambda \textbf{I}\right)^{-1}\textbf{Q}^H\textbf{G}_t$.


To sum up, a solution can be found through a series of line search, as summarized in Algorithm \ref{alg:beampattern_MM}.

\begin{algorithm}
\caption{MM-LineSearch algorithm for solving \eqref{prob:beampattern}}\label{alg:beampattern_MM}
\begin{algorithmic}
    \State \textbf{Input:} $L$, $P_a$, $\textbf{X}_c$, $\bm{\Delta}$, $\{d(\theta_u)\}_{u=1}^U$, $\epsilon_1$, $T_{max}$
    \State \textbf{Initialize:} Set \( t = 0 \); Randomly choose $\textbf{B}_t\in\mathcal{M}$.
    \While{$t \leq T_{\max}$ and $\Vert g(\textbf{B}_t)-g(\textbf{B}_{t-1})\Vert_F > \epsilon_1$}
        \State Find $\lambda$ that minimizes \eqref{eq:B_norm3} via line search
        \State Update $\textbf{B}_{t+1}$ by $\textbf{B}_{t+1}=\textbf{Q}\left(\bm{\Lambda}+\lambda \textbf{I}\right)^{-1}\textbf{Q}^H\textbf{G}_t$        
        \State $t=t+1$
    \EndWhile
    \State Formulate the transmit waveform by  $\textbf{X} = \textbf{X}_c + \bm{\Delta}\textbf{B}_t$
    \State \textbf{Output:} Transmit waveform $\textbf{X}$;
\end{algorithmic}
\end{algorithm}

\subsection{Complexity Analysis}

The total computational cost of the proposed MM-LineSearch algorithm depends on the number of iterations until convergence.
The eigenvalue decomposition of $\textbf{A}\textbf{A}^H$ costs $O\left(({N_t}-K)^3\right)$.
The computational complexity of the MM-LineSearch algorithm is dominated by the computation of $\textbf{G}_t$ and the line search.
The computation of $\textbf{G}_t$ is dominated by the computation of $\textbf{D}_t$, which costs $O(U({N_t}-K)^2L)$.
Each iteration of the bisection method requires examining \eqref{eq:B_norm3}, which costs $({N_t}-K)L$.
To sum up, the total computational cost for the MM-LineSearch algorithm can be expressed as
\begin{equation}
    O\left(({N_t}-K)^3+ n_{MM}U({N_t}-K)^2L+n_{1D}({N_t}-K)L\right)
\end{equation}
where $n_{MM}$ is the number of MM iterations and $n_{1D}$ is the total number of line search trials over all MM iterations.
It is difficult to characterize an upper bound for the number of iterations needed.
We empirically verified, however, that it converges\footnote{\BHL{The convergence behavior of \gls{mm} algorithms is well-studied. We refer the readers to \cite{hunterTutorialMMAlgorithms2004a} for a more detailed convergence analysis.}} within tens of iterations.


\section{
Application 2: Range Sidelobe Suppression
}

In this section, we address the problem of optimizing the correlation properties of the proposed waveform.
In radar sensing, the range sidelobe levels of a waveform impact the quality of range estimation and the capability to separate closely located targets.
In an \gls{isac} point of view, suppressing range sidelobes requires the use of explicit information about the data symbol realization, which hinders the use of traditional linear precoding.
To handle this, previous \gls{isac} work has adopted nonlinear precoding such as \gls{mui} minimization \cite{liuRangeSidelobeReduction2020} or symbol-level precoding \cite{lee2024constant}.
By contrast, our proposed method can build \gls{isac} signals from the combination of traditional linear precoding and deterministic sensing signals due to its ability to decouple sensing and communications.
In the following, we detail the range sidelobe suppression process for the proposed waveform.


\begin{figure}[!t]
\center{\includegraphics[width=.9\linewidth]{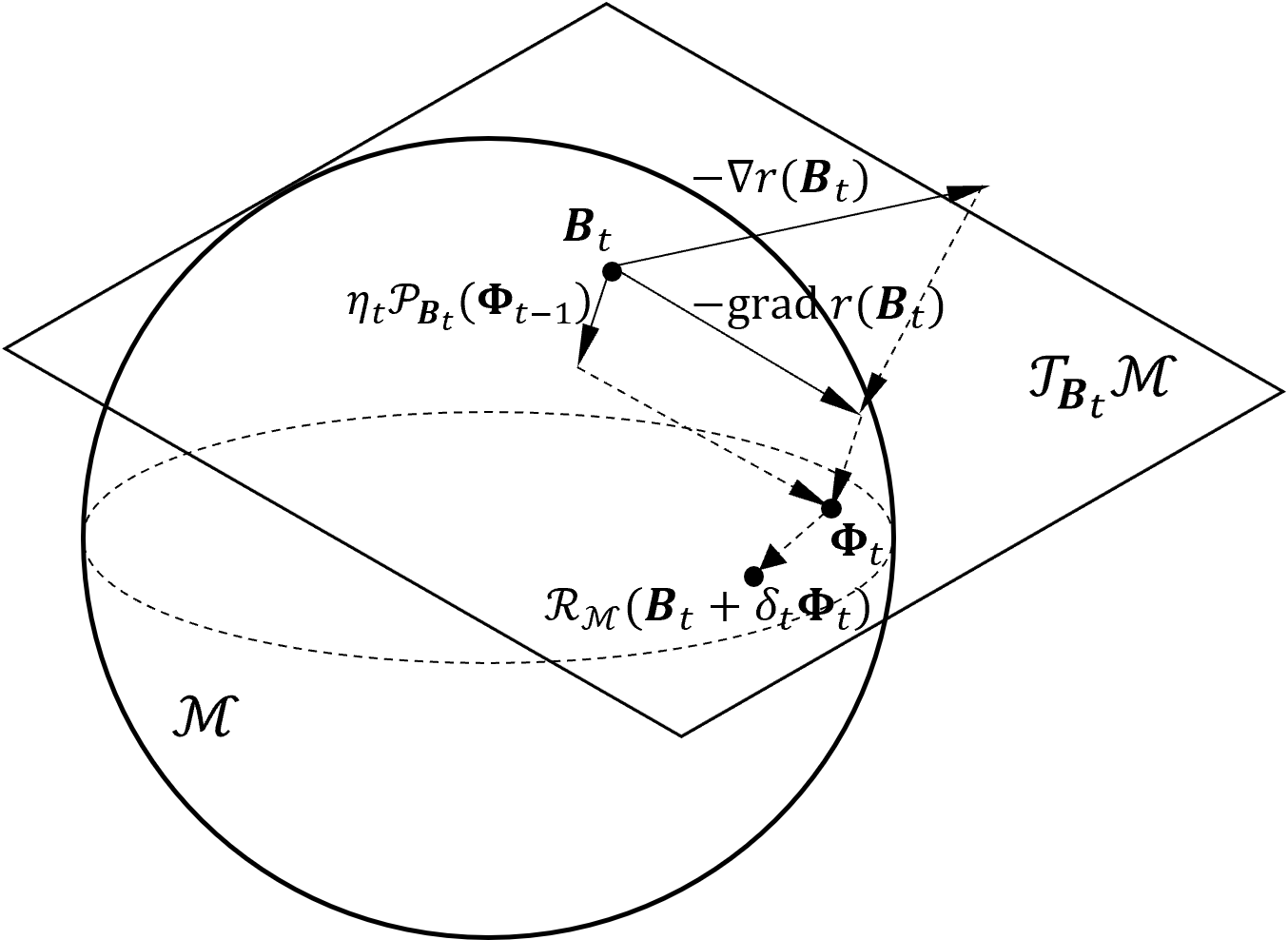}}
\caption{\small RCG algorithm on a complex hypersphere manifold. 
}
\label{fig:manifold}
\end{figure}

\subsection{Problem Formulation}
To distinguish multiple targets in different range bins, it is critical to minimize the correlation between the waveform and its shifted version.
\BHL{
The correlation matrix between the transmit waveform and its $\tau$-shifted version can be written as \cite{he2009designing}
\begin{equation}\label{eq:Omega}
    \bm{\Omega}_{\tau}=\textbf{X}\textbf{J}_{\tau}\textbf{X}^H=  \bm{\Omega}_{-\tau}^H.
\end{equation}
The diagonal entries of $\bm{\Omega}_{\tau}$ correspond to the autocorrelation for each antenna, whereas its off-diagonal entries represent the cross-correlations between two different transmit antennas.
The \gls{isl} quantifies the overall waveform correlation level and can be obtained by summing the squares of the cross- and autocorrelations over a set of range bins as \cite{he2009designing}
\begin{equation}\label{eq:ISL_raw}
    \begin{aligned}
    r(\textbf{B})
    &= \displaystyle\sum_{\substack{\tau=-P+1 \\ \tau\neq 0}}^{P-1} \Vert \bm{\Omega}_{\tau}\Vert_F ^2 \\
    &=
    \displaystyle\sum_{\substack{\tau=-P+1 \\ \tau\neq 0}}^{P-1} \left\Vert (\textbf{X}_c+\bm{\Delta}\textbf{B})\textbf{J}_{\tau}(\textbf{X}_c+\bm{\Delta}\textbf{B})^H\right\Vert_F^2
    \end{aligned}
\end{equation}
where $P$ is the largest range bin of interest
with $P-1\leq L$.
}

By minimizing \eqref{eq:ISL_raw}, the range resolution of the waveform can be improved.
Hence, the corresponding ISL minimization problem can be formulated as
\begin{mini}|s|
{\small\textbf{B}}{r(\textbf{B})}
{}{}
\addConstraint{\Vert \textbf{B}\Vert^2 = LP_a}{}
\label{prob:ISL}
\end{mini}
It is evident that the formulated problem \eqref{prob:ISL} is nonconvex due to the nonconvex fourth-order objective.
Nonetheless, thanks to the decomposibility of the \gls{sdisac} waveform, \eqref{prob:ISL} has no additional communication constraints, which facilitates reusing well-established existing solutions for \gls{isl} minimization in the \gls{mimo} radar literature with slight modification.

\subsection{Solution via Manifold Optimization}
We adopt a first-order Riemannian manifold optimization method \cite{absil2008optimization} that leverages the geometric structure of the feasible set.
\BHL{Riemannian manifold optimization follows a similar process to traditional gradient descent algorithms, except that the descent operation is done on a Riemannian manifold, as described in Fig. \ref{fig:manifold}.}
A Riemannian manifold is a manifold with tangent spaces that are equipped with a smoothly varying inner product \cite{kimImprovedMultiuserMIMO2009}. 
In the following, we detail the Riemannian manifold optimization process for solving the problem \eqref{prob:ISL}.

Let $\mathcal{M}=\{\textbf{B}\in\mathbb{C}^{{N_t}-K,L}:\Vert \textbf{B}\Vert_F^2= LP_a \}$ be the feasible set of \eqref{prob:ISL}.
$\mathcal{M}$ forms a complex hypersphere manifold \cite{absil2008optimization}, a set of points in $\mathbb{C}^{{N_t}-K,L}$ that are at distance $\sqrt{LP_a}$ from the origin.
Since a complex hypersphere is a Riemannian manifold, there exists a tangent space containing all the tangent vectors of $\mathcal{M}$ for any point $\textbf{B}\in \mathcal{M}$.
The tangent space of the manifold $\mathcal{M}$ is the hyperplane tangent to the complex hypersphere at a point $\textbf{B}\in\mathcal{M} $, which can be written as
\cite{absil2008optimization}
\begin{equation}
    \mathcal{T}_{\textbf{B}}\mathcal{M}=\{ \bm{\Gamma} \in \mathbb{C}^{L \times {N_t}-K}:\text{Re}\{\text{Tr}(\textbf{B}^H\bm{\Gamma})\}=0\}.    
    \label{eq:tangent_space}
\end{equation}


The added signal matrix can be updated on the tangent space of the manifold $\mathcal{M}$ at point $\textbf{B}_t$ as
\begin{equation}
    \textbf{B}_{t+1}=\textbf{B}_t +\delta_t\bm{\Phi}_{t},
    \label{eq:update_B}
\end{equation}
where $t$ is the iteration index, $\delta_t$ is the step size and $\bm{\Phi}_{t}$ is the descent direction.
The step size $\delta_t$ can be obtained by the Armijo backtracking line search method \cite{wright2006numerical}.
\BHL{
We consider the \gls{rcg} method for updating the descent direction \cite{absil2008optimization}, which computes the descent direction by a nonlinear combination of the Riemannian gradient at the current point and the descent direction in the previous iteration.}
A Riemannian gradient is the projection of a gradient of the cost function at a given point onto the tangent space, which accounts for the steepest descent direction on the tangent space.

\BHL{
We begin by computing the gradient of the ISL cost function \eqref{eq:ISL_raw} in the Euclidean space, which is given by \eqref{eq:grad_Euc}
\begin{equation}\label{eq:grad_Euc}
    \begin{aligned}
        &\nabla r(\textbf{B}) \\
        &=2\displaystyle\sum_{\substack{\tau=-P+1 \\ \tau\neq 0}}^{P-1} \Big(\bm{\Delta}^H(\textbf{X}_c+\bm{\Delta}\textbf{B})\textbf{J}_{\tau}(\textbf{X}_c+\bm{\Delta}\textbf{B})^H(\textbf{X}_c+\bm{\Delta}\textbf{B})\textbf{J}_{\tau}^T \\ & \qquad  \qquad +\bm{\Delta}^H(\textbf{X}_c+\bm{\Delta}\textbf{B})\textbf{J}_{\tau}^T(\textbf{X}_c  +\bm{\Delta}\textbf{B})^H(\textbf{X}_c+\bm{\Delta}\textbf{B})\textbf{J}_{\tau}\Big).  
    \end{aligned}    
\end{equation}
}
Next, the Riemannian gradient at point $\textbf{B}_t$ is computed by projecting \eqref{eq:grad_Euc} onto the tangent space \eqref{eq:tangent_space}, which can be written as
\begin{equation}\label{eq:Rim_grad}
    \text{grad}\ r(\textbf{B}_t)\triangleq \mathcal{P}_{\textbf{B}_{t}}\left(\nabla r(\textbf{B}_t)\right)     
\end{equation}
where $\mathcal{P}_{\textbf{B}_{t}}\left(\nabla r(\textbf{B}_t)\right)$ returns the projection of the gradient $\nabla r(\textbf{B})$ in the Euclidean space onto the tangent space $\mathcal{T}_{\textbf{B}_t}\mathcal{M}$, which is given by
\begin{equation}
     \mathcal{P}_{\textbf{B}_t}\left(\nabla r(\textbf{B}_t)\right) = 
     \nabla r(\textbf{B}_t) - \text{Re}\left\{\text{Tr}\left(\textbf{B}_t^H\nabla r(\textbf{B}_t)\right)\right\}\textbf{B}_t.
\end{equation}
In the \gls{rcg} framework, the descent direction is obtained by combining \eqref{eq:Rim_grad} and the descent direction in the previous iteration as \cite{wright2006numerical}
\begin{equation}
\label{eq:update_pi}
\begin{aligned}
    \bm{\Phi}_{t+1} &= -\text{grad}\ r(\textbf{B}_{t+1}) + \eta_{t+1} \mathcal{P}_{\textbf{B}_{t+1}}(\bm{\Phi}_t) \\
    &= -\nabla r(\textbf{B}_{t+1}) +\text{Re}\left\{\text{Tr}\left(\textbf{B}_{t+1}^H\nabla r\left(\textbf{B}_{t+1}\right)\right)\right\}\textbf{B}_{t+1} \\
    &\quad +\eta_{t+1}\left(\bm{\Phi}_{t} -\text{Re}\left\{\text{Tr}\left(\textbf{B}_{t+1}^H\bm{\Phi}_{t}\right)\right\}\textbf{B}_{t+1}\right),
\end{aligned}    
\end{equation}
where $\eta_{t+1}$ is the Polak-Ribiere combination coefficient, 
which is given by \cite{wright2006numerical}
\begin{equation}\label{eq:update_mu}
    \eta_{t+1} = 
    \frac{\langle \text{grad}\ r(\textbf{B}_{t+1}), \text{grad}\ r(\textbf{B}_{t+1})-\mathcal{P}_{\textbf{B}_{t+1}}\left(\text{grad}\ r(\textbf{B}_{t})\right)\rangle}{\langle \text{grad}\ r(\textbf{B}_{t}),\text{grad}\ r(\textbf{B}_{t})\rangle}.
\end{equation}
Note that the descent direction $\bm{\Phi}_t$ in the previous iteration is on a different tangent space $\mathcal{T}_{\textbf{B}_{t}}\mathcal{M}$ and thus should be projected onto the tangent space $\mathcal{T}_{\textbf{B}_{t+1}}\mathcal{M}$ in the current iteration, as presented in \eqref{eq:update_pi}.

Recall that the descent operation in \eqref{eq:update_B} was done on the tangent space $\mathcal{T}_{\textbf{B}_t}\mathcal{M}$ and thus the updated point may not fall into the feasible region $\mathcal{M}$.
Hence, the updated variable should be retracted into the feasible region $\mathcal{M}$.
To this end, we modify the update equation \eqref{eq:update_B} as 
\begin{equation}\label{eq:update_B_retract}
\textbf{B}_{t+1} = \mathcal{R}_{\mathcal{M}}(\textbf{B}_t+\delta_t \bm{\Phi}_t)
\end{equation}
where $\mathcal{R}_{\mathcal{M}}(\cdot)$ is a retraction operation for a complex hypersphere.
The retraction operation maps the updated variable to its nearest point to the feasible region as \cite{liuMUMIMOCommunicationsMIMO2018a}
\begin{equation}
    \mathcal{R}_{\mathcal{M}}(\textbf{B}) = \sqrt{LP_a}\frac{\textbf{B}}{\Vert \textbf{B}\Vert_F}.
\end{equation}
The described optimization process is iterated until convergence as summarized in Algorithm \ref{alg:ISL}.

\begin{algorithm}
\caption{RCG algorithm for ISL optimization}\label{alg:ISL}
\begin{algorithmic}[1]
\State \textbf{Input:} $L$, \( P_{a} \), \( \textbf{X}_c \), \( \bm{\Delta} \), $\epsilon_2$, $T_{max}$
    \State \textbf{Initialize:} Set \( t = 0 \) and $\bm{\Phi}_t=\textbf{0}$; Draw $\textbf{B}_t\in\mathcal{M}$.
    \While{$t \leq T_{\max}$ and $\Vert\text{grad}\ r(\textbf{B})\Vert_F > \epsilon_2$}
        \State Update $\textbf{B}_t$ using \eqref{eq:update_B}
        \State Update $\eta_{t+1}$ using \eqref{eq:update_mu}
        \State Update $\bm{\Phi}_{t+1}$ using \eqref{eq:update_pi}
        \State $t=t+1$
    \EndWhile
    \State Recover the transmit waveform by  $\textbf{X} = \textbf{X}_c + \bm{\Delta}\textbf{B}_t$
    \State \textbf{Output:} Transmit waveform $\textbf{X}$;
\end{algorithmic}
\end{algorithm}

\subsection{Complexity Analysis}
Each RCG iteration consists of the computation of \eqref{eq:grad_Euc}, \eqref{eq:update_pi} and \eqref{eq:update_B_retract}.
The computational cost of \eqref{eq:grad_Euc} is $O\left((({N_t}-K)N_tL+3N_tL^2+{N_t}^2L)(2P-2)\right)$.
In the case when ${N_t}\gg K$, this can be simplified to $O\left((N_tL^2+{N_t}^2L)(P-1)\right)$.
Calculating \eqref{eq:update_pi} requires calculating \eqref{eq:Rim_grad} and \eqref{eq:update_mu}.
The computation of \eqref{eq:Rim_grad} costs $O(({N_t}-K)^2L)$.
The computational cost of \eqref{eq:update_mu} is $O(({N_t}-K)L)$. 
Hence, calculating \eqref{eq:update_pi} takes $O(({N_t}-K)^2L)$.
 Finally, the computation of \eqref{eq:update_B_retract} is $O(({N_t}-K)L)$.
Combining the above results, the total computational cost can be expressed as 
\begin{equation}
    O\left(n_{RCG}\left((N_tL^2+{N_t}^2L)(P-1)+({N_t}-K)^2L\right)\right)
\end{equation}
where $n_{RCG}$ is the number of RCG iterations.
\BHL{Similar to our discussions  for the MM-LineSearch algorithm, the complexity of the RCG algorithm highly depends on the number of iterations. 
We confirmed the proposed RCG algorithm converges\footnote{We refer the readers to \cite{absil2008optimization} for a detailed convergence analysis.} within tens of iterations in our simulation environment. }

\section{Extension: Imperfect \gls{csi} Scenario}\label{sec:imperfect}
In this section, we extend the proposed method to the imperfect \gls{csi} scenario.
The proposed \gls{sdisac} method relies on the assumption of perfect \gls{csi} to project the additive signal onto the null space of the communication channel.
Yet, in practice, the \gls{csi} at the transmitter will not be perfect for practical reasons such as an estimation error or quantization.
Imperfect CSI destroys the orthogonality between the channel and the additive signal resulting in interference.
In the following, we address the problem of effective interference induced by imperfect \gls{csi}.

For a more general model, we adopt a transmit antenna correlation model, which can be expressed as \cite{kimMIMOSystemsLimited2011}
\begin{equation}\label{eq:corr_H}
    \textbf{H}=\textbf{H}^w\textbf{R}^{\frac{1}{2}}
\end{equation}
where $\textbf{H}^w$ is the Gaussian part of the channel matrix and $\textbf{R}\in\mathbb{C}^{{N_t}\times {N_t}}$ is a Hermitian positive semidefinite matrix representing the transmit spatial correlation.
We assume that the transmit spatial correlation $\textbf{R}$ changes much more slowly than the channel $\textbf{H}$ and thus can be estimated by the users and fed back to the BS \cite{sanguinettiMassiveMIMO202020}.

To model imperfect \gls{csi}, we follow the previous approach
\cite{kimMIMOSystemsLimited2011,chengwangAdaptiveDownlinkMultiuser2006}
that uses the Gauss-Markov model, which allows the matrix $\textbf{H}^w$ to be written as \cite{kimMIMOSystemsLimited2011}
\begin{equation}\label{eq:imperfect_csi}
    \begin{aligned}
        {\textbf{H}}^w
        =
        {\sqrt{1-\mu^2}\hat{\textbf{H}}^w}_{}+\mu\bm{\mathcal{E}}
    = \tilde{\textbf{H}}^w+\mu\bm{\mathcal{E}},
    \end{aligned}
\end{equation}
where $\hat{\textbf{H}}^w$ is the estimate of $\textbf{H}^w$ and $\bm{\mathcal{E}}$ is the error matrix.
We assume ${\textbf{H}}^w$, $\hat{\textbf{H}}^w$ and $\bm{\mathcal{E}}$ have i.i.d. Gaussian entries with unit variance.
For brevity, we denote $\sqrt{1-\mu^2}\hat{\textbf{H}}^w$ by $\tilde{\textbf{H}}^w$.
The parameter $\mu\in[0,1]$ accounts for the accuracy of the \gls{csi}.
$\mu=0$, $\mu=1$, and $0<\mu<1$ correspond to perfect \gls{csi} knowledge, no \gls{csi} knowledge, and partial \gls{csi} knowledge, respectively.
In general, $\mu$ can be modeled as a function of other system parameters, and is therefore assumed to be known.
For example, in minimum \gls{mse} (MMSE) estimation, $\mu$ is determined by the signal-to-noise ratio of the pilot signal \cite{nosrat2011mimo}. 
With \eqref{eq:imperfect_csi}, the channel matrix can be rewritten as \cite{nosrat2011mimo}
\begin{equation}
    \textbf{H}=\left(\tilde{\textbf{H}}^w+\mu\bm{\mathcal{E}}\right)\textbf{R}^{\frac{1}{2}}.
\end{equation}
Using the above channel model, the received communication signal can be expressed as
\begin{align}
    \textbf{Y}_c&=\textbf{H}\textbf{X}+\textbf{W}_c \\ \nonumber
    &=\tilde{\textbf{H}}^w\textbf{R}^{\frac{1}{2}}\textbf{X}+\mu\bm{\mathcal{E}}\textbf{R}^{\frac{1}{2}}\textbf{X} + \textbf{W}_c.
\end{align}
We construct the transmit signal as
\begin{equation}
    \textbf{X}=\textbf{X}_c+\bm{\Delta}\textbf{B}.
\end{equation}
 The projection matrix $\bm{\Delta}$ can be chosen by examining the \gls{svd} of $\tilde{\textbf{H}}^w\textbf{R}^{\frac{1}{2}}$, such that $\tilde{\textbf{H}}^w\textbf{R}^{\frac{1}{2}}\bm{\Delta}=\textbf{0}$.
With this in mind, the received signal can be rewritten as
\begin{align}    \textbf{Y}_c&=\tilde{\textbf{H}}^w\textbf{R}^{\frac{1}{2}}\textbf{X}+\mu\bm{\mathcal{E}}\textbf{R}^{\frac{1}{2}}\textbf{X} + \textbf{W}_c \\ \nonumber
    &=\tilde{\textbf{H}}^w\textbf{R}^{\frac{1}{2}}(\textbf{X}_c+\bm{\Delta}\textbf{B})+\mu\bm{\mathcal{E}}\textbf{R}^{\frac{1}{2}}(\textbf{X}_c+\bm{\Delta}\textbf{B}) + \textbf{W}_c \\ \nonumber
    &= \tilde{\textbf{H}}^w\textbf{R}^{\frac{1}{2}}\textbf{X}_c+\underbrace{\mu\bm{\mathcal{E}}\textbf{R}^{\frac{1}{2}}(\textbf{X}_c+\bm{\Delta}\textbf{B})}_{\text{Effective interference}} + \textbf{W}_c.
\end{align}
The effective interference energy is given by
\begin{align}
    &\mathbb{E}\left[\Vert\mu\bm{\mathcal{E}}\textbf{R}^{\frac{1}{2}}(\textbf{X}_c+\bm{\Delta}\textbf{B}) \Vert_F^2\right]  \\ \nonumber
    &=\mathbb{E}\left[\text{Tr}(\mu^2(\textbf{X}_c+\bm{\Delta}\textbf{B} )^H\textbf{R}^{\frac{1}{2}}\bm{\mathcal{E}}^H\bm{\mathcal{E}}\textbf{R}^{\frac{1}{2}}(\textbf{X}_c+\bm{\Delta}\textbf{B}) )\right] \\ \nonumber
    &=\text{Tr}\Big(\mu^2(\textbf{X}_c+\bm{\Delta}\textbf{B} )^H\textbf{R}^{\frac{1}{2}}\underbrace{\mathbb{E}\left[\bm{\mathcal{E}}^H\bm{\mathcal{E}}\right]}_{\textbf{I}_{N_t}}\textbf{R}^{\frac{1}{2}}(\textbf{X}_c+\bm{\Delta}\textbf{B}) \Big) \\ \nonumber
    &=
    \Vert\mu\textbf{R}^{\frac{1}{2}}(\textbf{X}_c+\bm{\Delta}\textbf{B}) \Vert_F^2.
\end{align}

To minimize the impact of the effective interference, we consider jointly minimizing it along with the radar objective.
In light of this, the beampattern matching problem \eqref{prob:beampattern} can be modified as 
\begin{mini}|s|
{\small\textbf{B}}{\omega g(\textbf{B})+ (1-\omega)\Vert\mu\textbf{R}^{\frac{1}{2}}(\textbf{X}_c+\bm{\Delta}\textbf{B}) \Vert_F^2}
{}{}
\addConstraint{\Vert \textbf{B}\Vert_F^2=LP_a}{}
\label{prob:beampattern_INR}
\end{mini}
where $\omega\in [0,1]$ is the weight.
The above problem can be still solved by the proposed MM-LineSearch algorithm.

\section{Discussions}
{
This subsection further discusses the benefits and potential issues of the proposed \gls{sdisac} waveform.
The proposed method decouples the generation of communication and radar waveforms, while these waveforms are jointly transmitted by the shared antenna array. 
This approach offers benefits, such as the reduced complexity and the enhanced compatibility with existing waveforms for both radar and communications.
This facilitates the implementation of \gls{isac} technology within legacy communication standards and infrastructure.

}


\subsection{Power and \Gls{dof} Trade-offs}
\BHL{ In the proposed framework, the communication and sensing signals share the transmit power, and the trade-off between the two functions is controlled by the parameter $\beta$.
Without loss of generality, in this paper, we consider allocating the minimum power required to satisfy the user \gls{qos} to communications.
Then the residual is used to design the added signal for optimizing radar performance.
As the power split ratio $\beta$ approaches unity, the power available for the added signal decreases, resulting in the communication signal dominating the overall transmit signal.}
The proposed method takes a two-step approach to designing transmit waveforms, where the added signal is designed on the null space of the communication channel, using the residual \gls{dof} after formulating communication beamforming.
The \gls{dof} of the added signal is ${N_t}-K$ assuming a full-rank channel matrix $\textbf{H}$.
In the massive \gls{mimo} scenario where ${N_t}\gg K$, ${N_t}\approx {N_t}-K$.
However, if ${N_t}$ is not sufficiently large compared to $K$, $K$ would impact performance due to the lower \gls{dof} of the added signal.



\subsection{Codebook-Based Null-Space Feedback}

In multi-user \gls{mimo} systems, the BS in general needs high-resolution \gls{csi} to control multi-user interference.
Therefore, in such systems, the BS would have enough information to find the null-space of the channel and formulate the matrix $\bm{\Delta}$.
By contrast, in some systems, communication users may not provide full \gls{csi} to the BS due to overhead constraints, even if they possess accurate channel information.
For example, in single-user \gls{mimo}, users might only feed back the most dominant direction in their channel.
In such cases, it becomes challenging to compute the null-space projection matrix $\bm{\Delta}$ at BS sides due to the limited \gls{csi} resolution.
A practical approach to this problem is to define a codebook $\mathcal{R}$ containing null-space basis vectors.
Users can select null-space basis vectors from the codebook $\mathcal{R}$ and feed back the corresponding indices to the BS.
Then, the BS can formulate the null-space projection matrix $\bm{\Delta}$ based on the received null-space feedback. 
Obviously, there is a trade-off between codebook resolution and feedback overhead, similar to classical codebook-based beamforming strategies \cite{loveLimitedFeedbackUnitary2005a}.

\subsection{Hybrid Random and Deterministic Signal}

Recent ISAC work revealed the trade-off between random and deterministic signal transmission in ISAC systems.
The proposed method exploits explicit symbol information to design the additive signal, which is referred to as a signal-dependent approach in the literature \cite{lee2024constant,liuDualFunctionalRadarCommunicationWaveform2021,liuRangeSidelobeReduction2020}. 
Signal-independent approaches focus on optimizing the second-order statistics of the waveform (e.g., covariance) \cite{liuJointTransmitBeamforming2020a,liuCramerRaoBoundOptimization2022}, which may be sensitive to the signal realization.
Conversely, signal-dependent approaches preserve waveform properties such as the beampattern, ambiguity function, and correlations.
Therefore, in general, signal-dependent approaches outperform signal-independent approaches \cite{luRandomISACSignals2024a}.
Nonetheless, incorporating explicit signal information into waveform design accompanies higher complexity, posing a new challenge.

To handle the randomness of the data symbols, recent ISAC work characterized or optimized the expectation of sensing metrics such as the \gls{lmmse} over random data and channel realizations \cite{luRandomISACSignals2024a}.
In contrast to these works, the proposed \gls{sdisac} waveform can explicitly optimize sensing metrics without having to optimize the expected sensing performance.



\subsection{Asynchronous Update of Sensing and Communication Signals}
Another important benefit of the decomposability offered by the proposed method is that communications and radar can be updated asynchronously.
As described earlier, the radar function is optimized by designing added signals on top of communication signals.
Thus, while the communication signal remains constant for a duration $T_c$, the added signal can be updated multiple times at shorter intervals $T_r$.
This enables two-timescale signal updates, improving the flexibility of implementing unified \gls{isac} waveforms.
Specifically, the two-timescale operation can be useful for scenarios where radar signals should be frequently updated such as space-time adaptive processing (STAP).
Since updating the beamforming matrix may cause overhead such as feedback and pilot transmission, the two-timescale \gls{isac} operation can assist in mitigating such overhead for communication systems.



\subsection{Correlation between Sensing and Communication Channels}

As with \gls{sdma}, the proposed method is suitable for situations where the target and users are well separated in the angular domain.
This is critical especially for spatial beampattern matching.
The proposed method projects the added signal onto the null-space of the communication channel to optimize sensing performance.
This implies that if the subspaces of the communication and sensing channels are highly aligned, then \gls{nsp} will result in significant signal energy loss.
This scenario is often called a highly-coupled scenario where concentrating signal energy in the common dominant direction of the communication and sensing channels is an effective solution, as in applications like vehicular networks \cite{liu2020radar}.
The subspace distance or the subspace overlap coefficient can be used to quantify the overlap between the sensing and communication subspaces \cite{xiongFundamentalTradeoffIntegrated2023}.

\BHL{
\subsection{Wideband Extension}
The proposed \gls{sdisac} method can be extended to wideband setup such as \gls{ofdm}.
For such an extension, added signals must be designed in a per-subcarrier manner due to frequency-selective fading.
The transmit signal in the wideband system can be modified as
\begin{equation}    \textbf{X}_q=\textbf{F}_q\textbf{S}_q+\bm{\Delta}_q\textbf{B}_q
\end{equation}
where $q$ is the subcarrier index.
The projection matrix $\bm{\Delta}_q$ should be calculated for every subcarrier, which may lead to increased complexity.
Wideband extension improves range resolution but also poses challenges like adaptive \gls{ofdm} radar signal processing, waveform design, and ambiguity function shaping \cite{senAdaptiveOFDMRadar2011,senOFDMMIMORadar2010,senAdaptiveDesignOFDM2010a}.
While this is an intriguing direction, we leave it as future work since it is beyond the scope of this paper.}


\subsection{Other Usage}

Although this paper focuses on \gls{isac} applications, the proposed \gls{sdisac} waveform can be applied to other problems, such as \gls{swipt} and \gls{papr} minimization.
\gls{swipt} aims to transmit information bits to communication users while transferring power to energy-harvesting devices.
For example, received energy at the harvesting devices can be maximized by solving the following problem:
\begin{maxi}|s|
{\small\textbf{B}}{\Vert\textbf{H}_h(\textbf{X}_c+\bm{\Delta}\textbf{B}) \Vert_{F}}
{}{}
\addConstraint{\Vert\textbf{B}\Vert_2 \leq LP_a}{}
\end{maxi}
where $\textbf{H}_h$ is the channel from the BS to the energy-harvesting devices.
In addition, the additive signal can be used to minimize the \gls{papr}, which can be expressed as 
\begin{mini}|s|
{\small\textbf{B}}{\Vert\textbf{X}_c+\bm{\Delta}\textbf{B} \Vert_{\infty}}
{}{}
\addConstraint{\Vert\textbf{B}\Vert_2 = LP_a}{}.
\end{mini}




\section{Simulation Results}


In this section, we evaluate our proposed waveform design method through rigorous simulations.
We use the following simulation setting for performance evaluation unless otherwise specified.
The transmit power is $P_t=1$ and communication noise variance is $\sigma^2_c=0.01$.
The BS transmit array is a uniform linear array (ULA) of ${N_t}=20$ antennas with half-wavelength spacing.
The codeword length was set to $L=64$.
We adopt the uncorrelated Rayleigh channel, where each entry of the channel matrix is drawn from $\mathcal{CN}(0,1)$.
The stopping criteria are set to $\epsilon_1=\epsilon_2=\SI{e-04}{}$.
Communication signals are generated using Algorithm \ref{alg:comm} to meet the minimum \gls{qos} requirements of the communication users.
The remaining power $P_a=P_t-P_c$ after constructing the communication signal is assigned to the added signal $\textbf{B}$ to enable radar functionality.

\begin{figure}
     \centering
     \begin{subfigure}[b]{0.45\textwidth}
         \centering
            \center{\includegraphics[width=.8\linewidth]{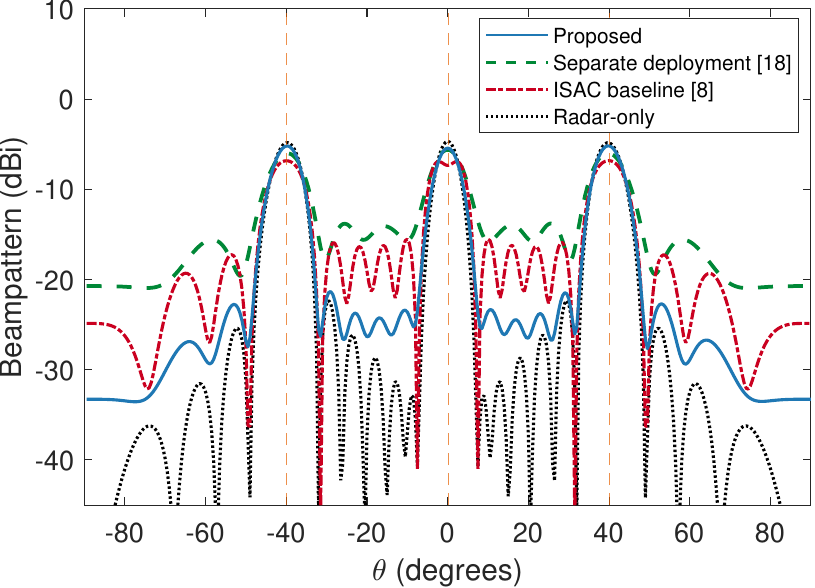}}
            \caption{ $\gamma_k=\SI{10}{\decibel}$.}
            \label{fig:beam_pattern_10dB}
         \end{subfigure}
     \hfill
     \begin{subfigure}[b]{0.45\textwidth}
         \centering
            \center{\includegraphics[width=.8\linewidth]{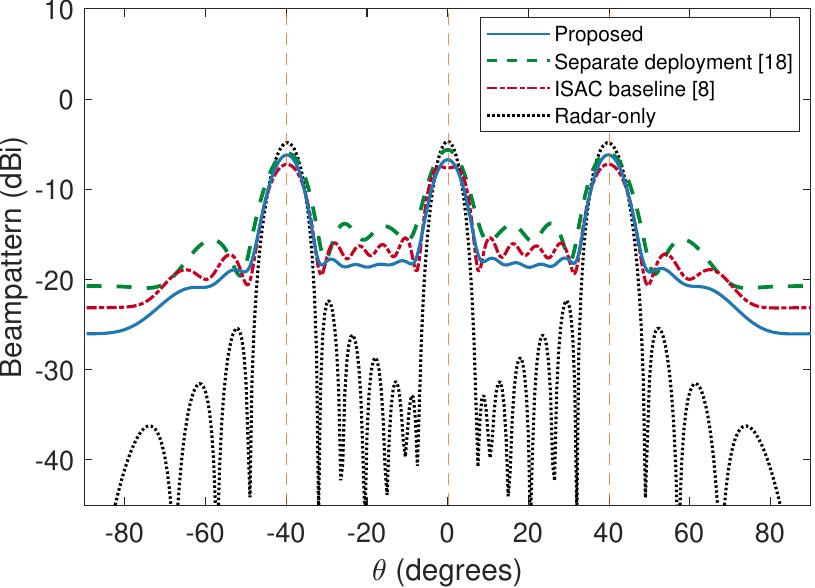}}
            \caption{ $\gamma_k=\SI{20}{\decibel}$.}
            \label{fig:beam_pattern_20dB}
         \end{subfigure}
     \caption{\small Beampatterns generated by the proposed method and baselines for $\gamma_k=\SI{10}{\decibel}$ and $\gamma_k=\SI{20}{\decibel}$.
    The number of users is $K=4$.
    The average power split ratios to meet $\gamma_k=\SI{10}{\decibel}$ and $\gamma_k=\SI{20}{\decibel}$ were $P_c = 0.049$ and $P_c=0.4936$, respectively.
Vertical lines correspond to three targets positioned at angles $\theta_1=-40^{\circ}$, $\theta_2=0^{\circ}$, and $\theta_3=40^{\circ}$.
     }
     \label{fig:Beam_Pattern}
\end{figure}

\begin{figure}[!t]
\center{\includegraphics[width=.75\linewidth]{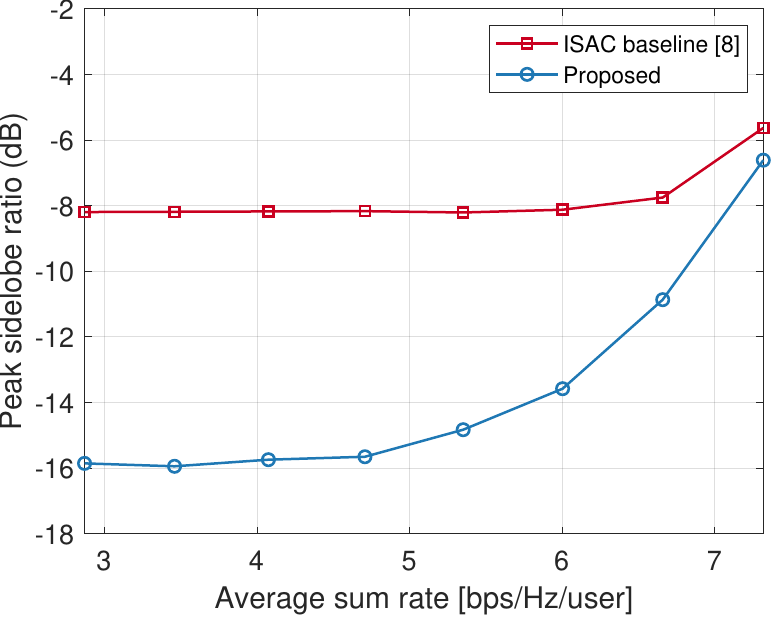}}
\caption{ \small 
Trade-off between the peak sidelobe ratio and average sum rate [bps/Hz/user] for the proposed scheme and \gls{isac} baseline \cite{liuJointTransmitBeamforming2020a}.
}
\label{fig:PSLR_vs_rate}
\end{figure}

\begin{figure}[!t]
\center{\includegraphics[width=.75\linewidth]{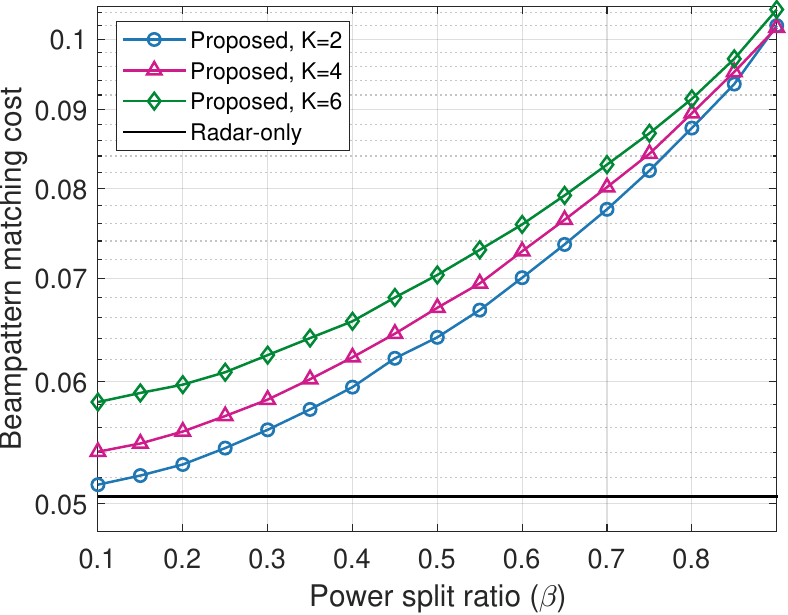}}
\caption{ \small 
Beampattern matching cost function versus power split ratio $\beta$ for $K=2$, $K=4$, and $K=6$.
}
\label{fig:MSE}
\end{figure}


\subsection{Beampattern Optimization}
In this subsection, we evaluate the beampattern matching performance of the proposed algorithm.
We use 1000 channel realizations to evaluate the average performance of the proposed algorithms unless otherwise specified.
We use $U=360$ uniformly sampled angles within a range $[-90^{\circ},90^{\circ}]$ to calculate the beampattern matching cost, which corresponds to the angle resolution of $0.5^{\circ}$.
This paper adopts a rectangular beampattern for the reference beampattern, which can be expressed as
\begin{equation}
    d(\theta)=
    \begin{cases}
    1,& \text{if } \theta_m - \Delta_{\theta}/2\leq\theta\leq  \theta_m + \Delta_{\theta}/2 \ \forall m=1,2,\dots,M\\
    0,              & \text{otherwise}
\end{cases}
\end{equation}
where $\Delta_{\theta}$ is the beam width.
We configured three targets, i.e., $M=3$, at angles $\theta_1=-40^{\circ}$ and $\theta_2=0^{\circ}$, and $\theta_3=40^{\circ}$
The beam width $\Delta_{\theta}$ is set to $10^{\circ}$.
For performance comparison, we use the algorithms in \cite{liuMUMIMOCommunicationsMIMO2018a} and  \cite{liuJointTransmitBeamforming2020a}.
\cite{liuMUMIMOCommunicationsMIMO2018a} synthesized a radar waveform for the separate deployment scenario where the radar and communication arrays are colocated yet not shared and the radar signal is projected onto the null space of the communication channel.
For the baseline \cite{liuMUMIMOCommunicationsMIMO2018a}, the numbers of radar and communication antennas are set to 14 and 6, respectively.
Conversely, \cite{liuJointTransmitBeamforming2020a} designed a beamformer using beampattern matching optimization subject to per-user \gls{sinr} constraints.
The radar-only baseline solves the same beampattern problem \eqref{prob:beampattern} without the communication signal, i.e., $\textbf{X}=\textbf{B}$ and $\bm{\Delta}=\textbf{I}_{N_t}$.

Fig. \ref{fig:beam_pattern_10dB} and Fig. \ref{fig:beam_pattern_20dB} show the beampatterns of the proposed method and baselines for two \gls{qos} constraints $\gamma_k=\SI{10}{\decibel}$, and $\gamma_k=\SI{20}{\decibel}$.
The number of users is set to $K=4$.
In both cases, the radar-only baseline outperforms \gls{isac} schemes in terms of spatial sidelobe level, providing a lower bound of \gls{isac} schemes.
Surprisingly, the proposed method achieved lower sidelobes and higher mainlobes than \gls{isac} baseline \cite{liuJointTransmitBeamforming2020a}.
The performance gain mainly comes from the different cost functions of the two schemes, indicating the effectiveness of minimizing \eqref{eq:beampattern_cost0}.
The baseline \cite{liuMUMIMOCommunicationsMIMO2018a} shows the worst performance among all schemes.
This is because \cite{liuMUMIMOCommunicationsMIMO2018a} splits antennas for communication and sensing, limiting the lower spatial resolution of communication and sensing.
In contrast, the proposed method exploits all antennas in the shared array, approximating beampattern with higher spatial \glspl{dof}.
 Notably, the overall sidelobe levels of the \gls{isac} schemes are higher for $\gamma_k=\SI{20}{\decibel}$ than for $\gamma_k=\SI{10}{\decibel}$ due to the sensing-communication trade-off, whereas \cite{liuMUMIMOCommunicationsMIMO2018a} remains unaffected by the \gls{qos} threshold as it uses a dedicated radar array.

Fig. \ref{fig:PSLR_vs_rate} evaluate the trade-off between the peak sidelobe ratio and the average sum rate for the proposed MM-LineSearch algorithm and baseline \cite{liuJointTransmitBeamforming2020a} when $K=4$.
The peak sidelobe ratio is a measure of the spatial resolution of the waveform, which is defined as the ratio of the peak sidelobe power to the peak mainlobe power.
For both schemes, it is evident that the peak sidelobe ratio tends to increase with the average sum rate, indicating a trade-off between sensing and communication performance.
The proposed method demonstrates a more favorable trade-off due to the introduced cost function \eqref{eq:beampattern_cost0}, which aligns with the results shown in Fig. \ref{fig:beam_pattern_10dB} and Fig. \ref{fig:beam_pattern_20dB}.
 This suggests that the proposed approach achieves higher spatial resolution without requiring a joint optimization of sensing and communications.
 The performance gap between the proposed scheme and baseline diminishes as the average sum rate increases.
 This convergence occurs because the transmit signal becomes dominated by the communication signal with less available power allocated to the added signal.
The peak sidelobe ratios of the two schemes tend to converge under these power-limited conditions.

 We assess the impact of the power split ratio $\beta=P_c/P_t$ on the beampattern matching performance.
Fig. \ref{fig:MSE} plots the beampattern matching cost defined in \eqref{eq:beampattern_cost0} with increasing power split ratio for $K=2$, $K=4$, and $K=6$.
The radar-only scheme provides a lower bound for the proposed method, and as the power split ratio decreases, the proposed scheme gradually converges to the radar-only performance.
Nonetheless, the cost function is higher with higher $K$ due to the lower \gls{dof}.
It is shown that the beampattern MSE increases as the power split ratio approaches unity.
As the power split ratio increases, the power budget for the added signal decreases as a trade-off, leading to inferior beampattern approximation.
Moreover, the cost function is higher for the higher number of users because the rank of the projection matrix diminishes with the increasing number of users. 
The difference among the three cases becomes less prominent at high power split ratio.
This suggests that at high power split ratio, the residual power is a more dominant factor to performance than the number of available basis vectors in beampattern approximation.

\begin{table}[!t]
        \centering
        \renewcommand{\arraystretch}{1.4}
        \caption{CPU time comparison of MM-LineSearch and \gls{isac} baseline \cite{liuJointTransmitBeamforming2020a} for different values of $K$.}
        \begin{tabular}{|l|c|c|c|c|c|}
            \hline
            \textbf{Method} & $K=1$ & $K=2$ & $K=3$ & $K=4$ & $K=5$ \\ \hline \hline
            MM-LineSearch & 0.56& 0.54& 0.47 & 0.42 & 0.35\\ \hline
            SDR \cite{liuJointTransmitBeamforming2020a}& 2.67& 2.75& 2.81&2.97 &3.06 \\ \hline
        \end{tabular}
        \label{tab:comparison}
\end{table}

In Table \ref{tab:comparison}, we compare the CPU time of the proposed MM-LineSearch algorithm with the \gls{isac} baseline from \cite{liuJointTransmitBeamforming2020a}. The results indicate that MM-LineSearch achieves significantly faster convergence. 
The CPU time for \cite{liuJointTransmitBeamforming2020a} grows rapidly with increasing $K$, as it requires designing $K$ covariance matrices, resulting in an increase in the number of optimization variables proportional to $K^6$. 
In contrast, the MM-LineSearch algorithm maintains efficiency as the number of users grows, even accelerating since the matrix size decreases with larger $K$.
\BHL{This is because the size of matrix $\textbf{B}$, $(N_t-K)\times L$ decreases as $K$ grows.
However, this comes at the cost of a lower \gls{dof} as the rank of the projection drops with the increase in $K$. }

\begin{figure}[!t]
\center{\includegraphics[width=.75\linewidth]{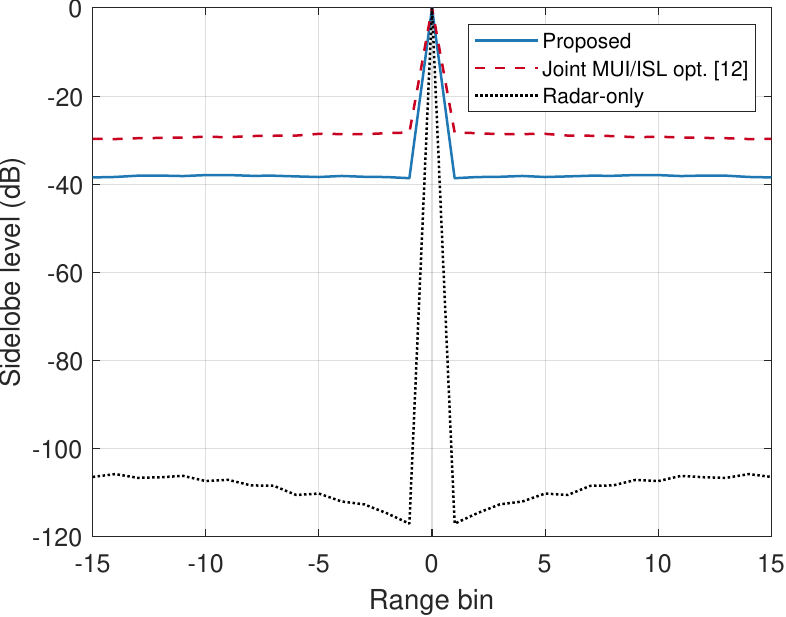}}
\caption{ \small Range sidelobes of the proposed \gls{sdisac} waveform, and baseline waveforms for $K=4$.
The achievable sum rates for the proposed \gls{sdisac} waveform and the joint \gls{mui}/\gls{isl} optimization scheme \cite{liuRangeSidelobeReduction2020} were $11.4791$ and $10.7713$ [bps/Hz], respectively.
}
\label{fig:sidelobe}
\end{figure}


\begin{figure}[!t]
     \centering
     \begin{subfigure}[b]{0.24\textwidth}
         \centering
            \center{\includegraphics[width=.95\linewidth]{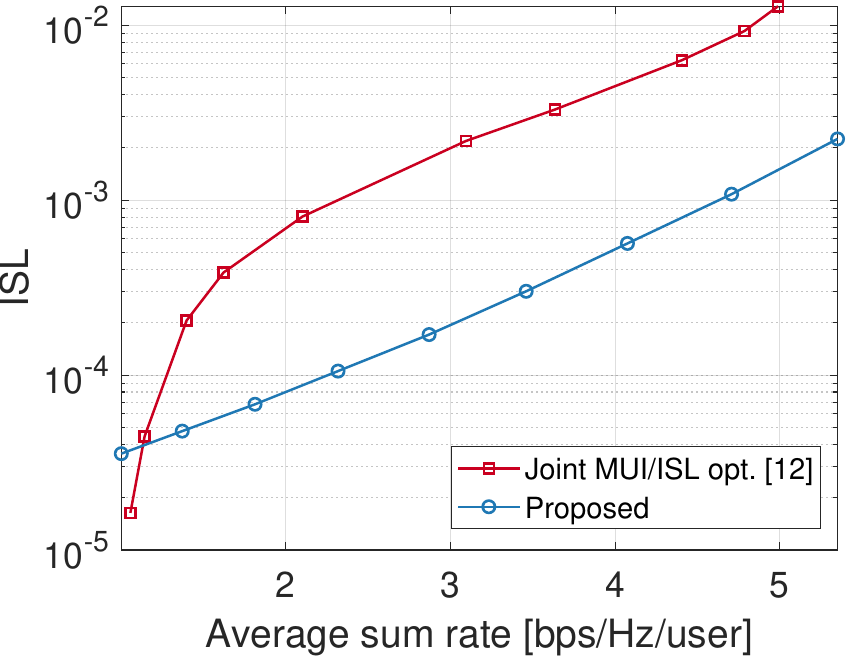}}
            \caption{ $K=2$.}
            \label{fig:ISL_vs_rate_K2}
         \end{subfigure}
     \hfill
     \begin{subfigure}[b]{0.24\textwidth}
         \centering
            \center{\includegraphics[width=.95\linewidth]{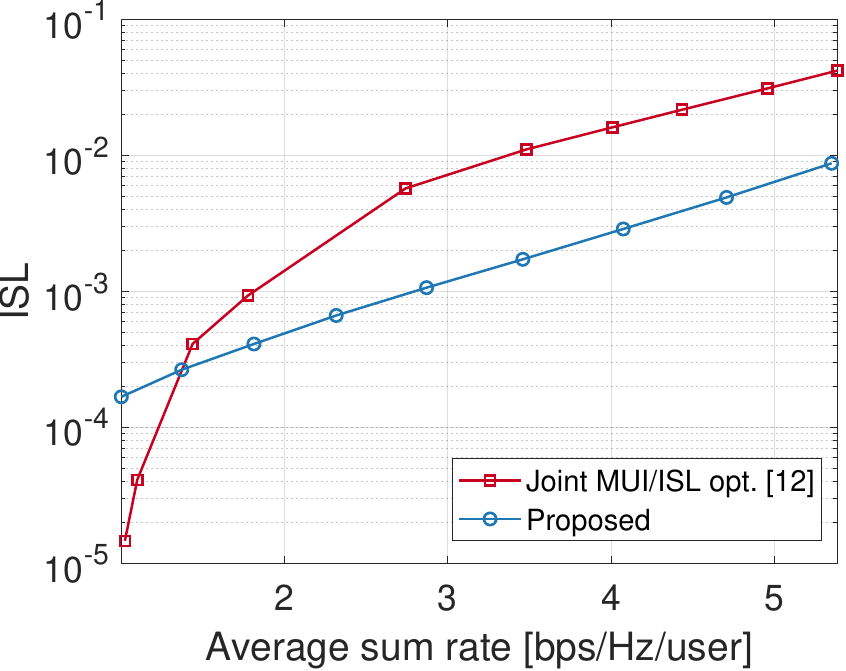}}
            \caption{ $K=4$.}
            \label{fig:ISL_vs_rate_K4}
         \end{subfigure}
     \caption{\small Trade-off between \gls{isl} and average sum rate [bps/Hz/user] for $K=2,4$.
     }
     \label{fig:ISL_vs_rate}
\end{figure}

\begin{figure}[!t]
\center{\includegraphics[width=.75\linewidth]
{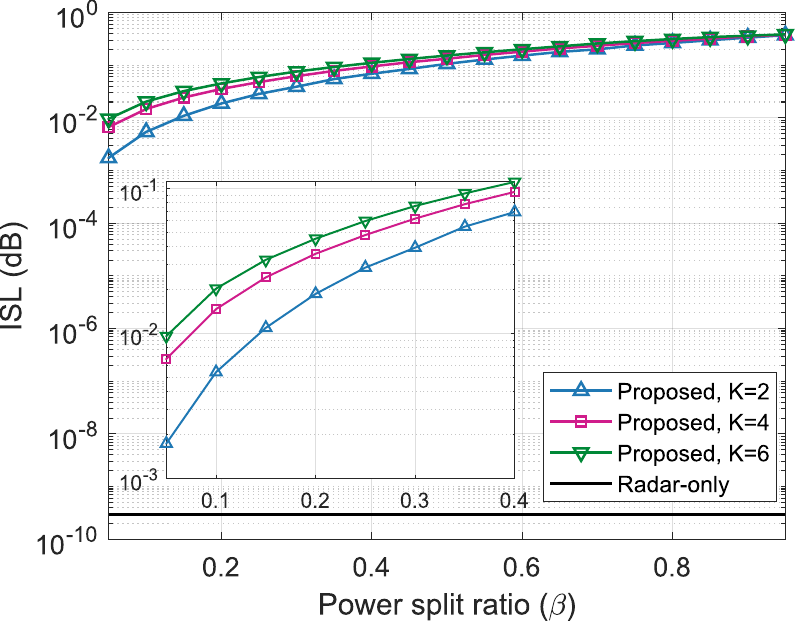}}
\caption{ \small 
\gls{isl} versus power split ratio for $K=2,4,6$.
}
\label{fig:ISL}
\end{figure}

\subsection{Range Sidelobe Suppression}
In this subsection, we assess the range sidelobe suppression performance of the \gls{sdisac} waveform.
For this simulation, we set $P=16$.
We use the algorithm in \cite{liuRangeSidelobeReduction2020} as a baseline, which jointly optimizes \gls{mui} and ISL.
For a fair comparison, we configured the weight parameters for the baseline for which the two \gls{isac} schemes have comparable sum rates.
Each scheme is averaged over $200$ Monte-Carlo simulations.

Fig. \ref{fig:sidelobe} evaluates the sidelobe suppression performance of the proposed \gls{sdisac} waveform and baselines for $K=4$.
\BHL{The sidelobe level on the y-axis represents the sum of the cross- and autocorrelation power for each range bin, i.e., $\Vert \bm{\Omega}_{\tau}\Vert_F^2$ in \eqref{eq:ISL_raw}.
}
For fair comparison, we set the sum rates of the proposed scheme and the baseline \cite{liuRangeSidelobeReduction2020} comparably to $11.4791$ and $10.7713$ [bps/Hz], respectively.
It can be observed that the \gls{sdisac} waveform outperforms the baseline \cite{liuRangeSidelobeReduction2020} by approximately \SI{10}{\decibel} despite its higher sum rates, showing the superior trade-off of the proposed method.
This is mainly because the \gls{sdisac} waveform allows for linear precoding whereas \cite{liuRangeSidelobeReduction2020} relies on \gls{mui} minimization.
Minimizing \gls{mui} is not always optimal in terms of sum rate as it does not allow the scaling or rotation of the constellation on the receiver side, limiting the communication performance.
The radar-only waveform shows significantly lower range sidelobes than the \gls{isac} counterparts, indicating the trade-off between sensing and communications.

Fig. \ref{fig:ISL_vs_rate_K2} and Fig. \ref{fig:ISL_vs_rate_K4} compare the trade-offs between ISL and sum rate for the proposed \gls{sdisac} waveform and baseline for $K=2$ and $K=4$.
In both cases, the proposed method outperforms the baseline except for the low sum rate region, which aligns with the result in Fig. \ref{fig:sidelobe}.
As discussed earlier, this is mainly due to the suboptimal nature of the \gls{mui} minimization approach compared to linear precoding. 
Notably, the slope of the proposed method's curve is not as steep as that of the baseline in the low sum rate region, resulting in the cross-over of the two curves.
This is because the rank of the added signal is ${N_t}-K=18$, which is lower than the rank of the baseline waveform that has rank ${N_t}=20$.
In other words, the impact of the lower rank signal is higher in the low sum rate region.
This can be seen as a trade-off for additional benefits of the \gls{sdisac} waveform.
However, it is important to note that the cross-over point is lower when $K=2$ than when $K=4$.
This implies that the impact of the lower-rank added signal diminishes as the user-number-to-antenna-number ratio decreases, i.e., $K/{N_t}$ goes to zero.

Fig. \ref{fig:ISL} plots the ISL of the proposed \gls{sdisac} waveform with increasing power split ratio values.
In line with the previous results, the radar-only scheme presents a lower bound of \gls{isl} performance.
The gap between the proposed and radar-only schemes is quite substantial.
This motivates further investigation, which we designate as future work. 
Moreover, the \gls{isl} grows with the increasing power split ratio due to the power trade-off. 
The \gls{isl} is also impacted by the number of users, $K$.
This is because the \gls{dof} is lower with higher $K$, making range sidelobe suppression more challenging.  
At high power split ratio, all three cases converge closely.
This suggests that the power split ratio dominates range suppression performance under high power split ratio conditions.

\subsection{Performance under Imperfect \gls{csi}}

We evaluate the beampattern matching performance of the proposed method under imperfect \gls{csi}.
To handle the effective interference induced by the channel mismatch, we jointly minimize the radar objective and effective interference using the MM-LineSearch algorithm, as described in Section \ref{sec:imperfect}.
We adopt the exponential correlation matrix whose entry in the $i$th row and $j$th column is $R_{i,j}=\rho^{|i-j|}$ with $\rho\in \mathbb{C}$ and $|\rho|\in [0,1]$ \cite{choiDownlinkTrainingTechniques2014}.
In this paper, we set $\rho=0.6$ \cite{choiDownlinkTrainingTechniques2014}.

Fig. \ref{fig:SumRate_imperfect} shows the achievable sum rate with increasing weights $\omega$.
We consider power split ratio values $\beta=0.2,0.5$ and CSI accuracy parameters $\mu=0.1,0.3$ in \eqref{eq:imperfect_csi}.
For all cases, the achievable sum rate decreases as the weight $\omega$ increases.
This is because a higher $\omega$ corresponds to less emphasis on interference suppression, leading to increased effective interference.
The achievable sum rate is significantly lower when $\mu=0.3$ than when $\mu=0.1$ due to larger effective interference.
Moreover, we can see interference suppression has a greater impact on the sum rate when $\mu=0.3$ than when $\mu=0.1$.
Moreover, for the same $\mu$, interference suppression was more critical at lower power split ratio.
This is because the added signal becomes relatively larger than the communication signal at low power split ratio, resulting in stronger effective interference.

Fig. \ref{fig:PSLR_imperfect} shows the peak sidelobe ratio with increasing values of weight $\omega$.
For all cases, the peak sidelobe ratio decreases as the weight $\omega$ increases, which shows the trade-off between beampattern matching and interference suppression.
It can be seen the peak sidelobe ratio is higher when the CSI accuracy was lower.
This implies that suppressing interference becomes challenging with higher $\mu$, resulting in a less favorable trade-off between radar sensing and interference suppression.

\begin{figure}[!t]
\center{\includegraphics[width=.75\linewidth]{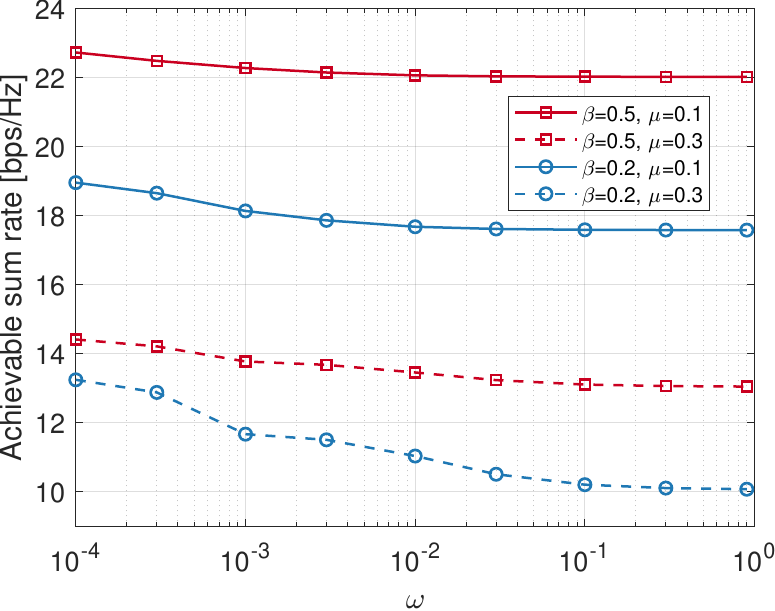}}
\caption{ \small Achievable sum rate versus weight $\omega$ for
$\beta=0.2,0.5$ and $\mu=0.1,0.3$.
\label{fig:SumRate_imperfect}
}

\end{figure}
\begin{figure}[!t]
\center{\includegraphics[width=.75\linewidth]{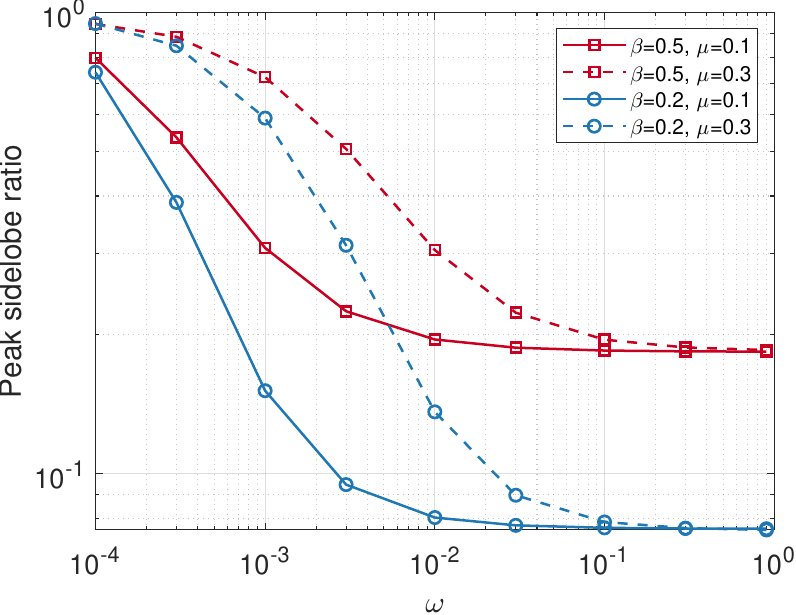}}
\caption{ \small Peak sidelobe ratio versus weight $\omega$ for
$\beta=0.2,0.5$ and $\mu=0.1,0.3$.
}
\label{fig:PSLR_imperfect}
\end{figure}

\section{Conclusion and Future Work}

This paper investigated a novel waveform structure, \textit{the \gls{sdisac} waveform}, which leverages the \gls{nsp} and \gls{iv} methods.
Specifically, we demonstrated a method of superimposing sensing signals onto communication signals without affecting the communication channel. 
The proposed \gls{sdisac} waveform decouples communication and radar waveform design tasks, offering benefits such as lower complexity.
We presented applications of the proposed \gls{sdisac} waveform, such as beampattern and correlation optimization.
Moreover, we discussed the impact of imperfect \gls{csi} on the performance of the proposed waveform and strategies to address it.
We evaluated the performance of the \gls{sdisac} waveform through simulations.
Simulation results showed that the proposed method can achieve comparable performance to the joint waveform design of radar and communications, despite additional benefits such as complexity and time flexibility.
Future work could extend the proposed method to MIMO orthogonal frequency division multiplexing (OFDM) systems.  

\appendices

\section{Solution to Problem \eqref{prob:comm}}
By letting $\textbf{T}_k=\textbf{f}_k\textbf{f}_k^H$ and $\textbf{T}=\sum_{k=1}^{K}\textbf{T}_k$, the problem \eqref{prob:comm} can be transformed into a semidefinite program as 
\begin{equation}
    \begin{aligned}
        & \underset{\textbf{T},\{\textbf{T}_k\}_{k=1}^K}{\text{min}}
        & & \text{Tr}(\textbf{T}) \\
        & \text{s.t.}
        & & \textbf{T} = \sum_{k=1}^{K} \textbf{T}_k, \ \text{Rank}(\textbf{T}_k)=1, \ \forall k \\
        & & & \frac{\textbf{h}^H_k \textbf{T}_k \textbf{h}_k}{\sum_{k'=1,k' \neq k}^K \textbf{h}^H_k \textbf{T}_{k'} \textbf{h}_k + L\sigma_c^2} \geq \gamma_k, \ \forall k.
    \end{aligned}
    \label{eq:optimization_problem}
\end{equation}
Using the \gls{sdr} technique, the above problem can be reformulated as 
\begin{equation}
    \begin{aligned}
        & \underset{{\textbf{T},\{\textbf{T}_k\}_{k=1}^K}}{\text{min}}
        & & \text{Tr}(\textbf{T}) \\
        & \text{s.t.}
        & & \textbf{T} - \sum_{k=1}^{K} \textbf{T}_k \succeq 0 \\
        & & & \frac{\textbf{h}^H_k \textbf{T}_k \textbf{h}_k}{\sum_{k'=1,k' \neq k}^K \textbf{h}^H_k \textbf{T}_{k'} \textbf{h}_k + L\sigma_c^2} \geq \gamma_k, \ \forall k.
    \end{aligned}
    \label{eq:optimization_problem_r}
\end{equation}
The relaxed problem is convex, which can be solved with numerical tools.
The beamforming vectors $\{\textbf{f}_k\}_{k=1}^K$ can be recovered from the obtained optimal covariances $\textbf{T}$ and $\{\textbf{T}_k\}_{k=1}^K$, as detailed in \cite{liuJointTransmitBeamforming2020a}.
The power minimization algorithm is summarized in Algorithm \ref{alg:comm}.
\begin{algorithm}
\caption{SDR algorithm for solving \eqref{prob:comm}}\label{alg:comm}
\begin{algorithmic}
    \State \textbf{Input:} $P_a$, $\textbf{X}_c$, $\bm{\Delta}$, $\{d(\theta_u)\}_{u=1}^U$
    \State \textbf{Steps:}
    \State Obtain $\textbf{T}$ and $\{\textbf{T}_k\}_{k=1}^K$ by solving problem \eqref{eq:optimization_problem_r}
    \State Recover $\textbf{B}$ from $\textbf{T}$ and $\{\textbf{T}_k\}_{k=1}^K$
    \State $\textbf{X} \gets \textbf{X}_c + \bm{\Delta}\textbf{B}$
    \State \textbf{Output:} Transmit waveform \textbf{X}
\end{algorithmic}
\end{algorithm}


\bibliographystyle{IEEEtran}
\bibliography{IEEEabrv,references}

\end{document}